\documentclass[final, 5p, times, twocolumn, authoryear]{elsarticle}

\journal{Econometrics and Statistics}

\usepackage{epsfig}
\usepackage{amssymb,amsmath,amsthm,amsfonts}
\usepackage[T1]{fontenc}
\usepackage{mathrsfs}
\usepackage{listings} 
\usepackage{float}
\usepackage{relsize}
\usepackage{booktabs}
\usepackage{enumitem}
\usepackage{multirow}
\usepackage[toc,page]{appendix}
\usepackage{mathdots}
\usepackage{graphicx}
\usepackage{graphics}
\usepackage{caption}
\usepackage{subcaption}
\numberwithin{equation}{section}
\usepackage{xcolor}
\usepackage{arydshln}
 
\newcommand{\N}{{\mathord{\mathbb N}}}
\newcommand{\NORM}{{\mathcal{N}}}
\newcommand{\Z}{{\mathord{\mathbb Z}}}
\newcommand{\R}{{\mathord{\mathbb R}}}
\newcommand{\Prob}{{\mathord{\mathbb P}}}

\newcommand{\F}{{\mathord{\mathcal F}}}

\newcommand{\K}{{\mathord{\mathcal K}}}
\newcommand{\E}{{\mathord{\mathbb E}}}
\newcommand{\norm}[1]{\lVert#1\rVert}
\newcommand{\eps}{\epsilon}
\DeclareMathOperator*{\argmin}{arg\,min}

\newcommand{\bigO}{\mathcal{O}}

\newtheorem{thm}{Theorem}[section]

\newtheorem{cor}{Corollary}[section]

\theoremstyle{definition}

\newtheorem{assumpW}{Assumption}

\usepackage[ruled,vlined]{algorithm2e}
\makeatletter
\newcommand{\removelatexerror}{\let\@latex@error\@gobble}
\makeatother

\begin{document}


\begin{frontmatter}


\title{AdaVol: An Adaptive Recursive Volatility Prediction Method}

\author[address1]{Nicklas Werge\corref{c1}}
\ead{nicklas.werge@upmc.fr}
\author[address1]{Olivier Wintenberger}
\ead{olivier.wintenberger@upmc.fr}
\address[address1]{LPSM, Sorbonne Universit\'e, 4 place Jussieu, 75005 Paris, France}
\cortext[c1]{Corresponding author}

\begin{abstract}
Quasi-Maximum Likelihood (QML) procedures are theoretically appealing and widely used for statistical inference.
While there are extensive references on QML estimation in batch settings, it has attracted little attention in streaming settings until recently.
An investigation of the convergence properties of the QML procedure in a general conditionally heteroscedastic time series model is conducted, and the classical batch optimization routines extended to the framework of streaming and large-scale problems.
An adaptive recursive estimation routine for GARCH models named AdaVol is presented.
The AdaVol procedure relies on stochastic approximations combined with the technique of Variance Targeting Estimation (VTE).
This recursive method has computationally efficient properties, while VTE alleviates some convergence difficulties encountered by the usual QML estimation due to a lack of convexity.
Empirical results demonstrate a favorable trade-off between AdaVol's stability and the ability to adapt to time-varying estimates for real-life data.
\end{abstract}



\end{frontmatter}


\section{Introduction}


Time series analysis has attracted much attention in the last three decades.
A central aspect of time series analysis is modeling heteroscedasticity of the conditional variance, e.g., volatility clustering in financial time series. 
Some well-known models incorporating this feature are the AutoRegressive Conditional Heteroscedasticity (ARCH) model and the Generalized ARCH (GARCH) model introduced by \cite{engle1982} and \cite{bollerslev1986}, respectively.
Many reasons can explain these models' success; they constitute a stationary time series model with a time-varying conditional variance, and secondly, they may model time series with heavier tails than the Gaussian ones, which often occurs in financial time series.

Quasi-Maximum Likelihood (QML) estimation is widely used for statistical inference in GARCH models due to their appealing theoretical nature and tolerance to overdispersion, which is often observed in empirical data.
This paper studies the Quasi-Maximum Likelihood Estimator (QMLE) for the broader class of conditionally heteroscedastic time series models of multiplicative form given by
\begin{align} \label{CHTSM}
X_t = h_t(\theta_0) Z_t, \quad t \in \Z,
\end{align}
where $\theta_0$ is the true underlying parameter vector, $(Z_t)$ is a sequence of i.i.d. random variables with $\E [Z_0] =0$ and $\E [Z^2_0] =1$, and the (non-negative) volatility process $(h_t)_{t \in \Z}$ is defined as
\begin{align} \label{CHTSM_h}
h_t (\theta) = g_\theta \big( X_{t-1}, \dots, X_{t-p}, h_{t-1}(\theta) , \dots, h_{t-q}(\theta)  \big), \quad p,q \geq 0.
\end{align}
Suppose that the parameter set $\Theta \subset \R^d$ and $\{ g_\theta \vert \theta \in \Theta \}$ denotes the (finite) parametric family of non-negative functions on $\R^p \times [0, \infty)^q$ satisfying certain regularity conditions. 
We also require that $h_t$ is $\F_{t-1}$-measurable for all $t \in \Z$, where $\F_t = \sigma ( Z_k : k \leq t)$ denotes the $\sigma$-field generated by the random variables $\{ Z_k : k \leq t \}$.

The stability of model \eqref{CHTSM}-\eqref{CHTSM_h} is accomplished under the assumption that $g_\theta$ is a contraction. 
This condition is a random Lipschitz coefficient condition, where the Lipschitz coefficient has a negative logarithmic moment. 
The notion of contractivity is clarified in \cite{straumannmikosch2007} where they study QML inference of general conditionally heteroscedastic models with emphasis on the approximation $(\widehat{h}_t)$ of the stochastic volatility $(h_t)$.

QML estimation of the parameters in the class of conditionally heteroscedastic time series models has been studied frequently in recent years, see e.g., \cite{berkes2003}, \cite{francq2004}, \cite{straumannmikosch2007}, and \cite{wintenberger2013}. 
However, all these references consider iterative estimation, where one assembles a batch of data and afterward performs the statistical inference.
Thus, one evaluates an objective function consisting of a sum of $n$ loss terms.
Each iteration would then have a cost of $\bigO(nd)$, making the recursion cost $\bigO(mnd)$, where $m$ is the number of iterations.
As the amount of data grows, these optimizers become prohibitively expensive and increasingly computationally inefficient.
Moreover, iterative optimizers become unsuitable for streaming settings where we are modeling and predicting data as they arrive.

Many financial practices, such as banks, asset managers, and financial services institutes, find themselves estimating thousands of volatility models every day for risk and pricing purposes.
In addition, the sampling of financial time series is increasingly at high frequency.
Therefore, recursive procedures must undoubtedly be advantageous since one only processes observations once.
In recursive QML estimation, we update the previous QML estimate with the new observations at time $t$ in order to produce the QML estimate of the parameters at time $t$.

Thus, in modern statistical analysis, it is becoming increasingly common to work with streaming data where one observes only a group of observations at a time. 
Naturally, this has led to an expanded interest in time-scalable recursive estimation procedures with a cost of only $\bigO(d)$ computations per recursion, e.g., see \cite{bottou2007}.
However, there has only been given a little amount of attention to recursive estimation in conditionally heteroscedastic time series models.

\cite{dahlhaus2007} presented a recursive method for estimating the parameters of an ARCH process.
Under sufficient assumptions on the underlying process, \cite{aknouche2006} showed consistency of their recursive least squares method for GARCH processes, and \cite{kierkegaard2000} also developed a recursive estimation method for GARCH processes supported by empirical evidence.
Convergence analysis of the recursive QML estimator for GARCH processes based on stochastic approximations with Markovian dynamics using a resetting mechanism has been previously presented (\cite{gerencser2010}).
A self-weighted recursive estimation algorithm for GARCH models was proposed in \cite{cipra2018} with a robustification in \cite{hendrych2018}. 
However, none of the above references mention problems with convexity or address the obstacles that may occur when the true parameter $\theta_{0}$ is close to the boundary of the parameter space.

The difficulty of estimating time-varying parameters of statistical models increases in the setting of streaming data. 
To sustain computational efficiency and be adaptive to changes in the estimates, one may decrease the number of observations in each iteration in the optimization procedure, which may decrease the stability of the statistical inference. 
We propose a natural adaptation of the QML method, relying on stochastic approximations combined with the Variance Targeting Estimation (VTE) technique, which we call AdaVol.
This recursive method is time-scalable and memory-efficient, as it only requires the previous estimate to process new observations, and it only needs to treat the observations once.
We present empirical evidence that AdaVol achieves a favorable trade-off between adaptability and stability.

The rest of the paper is organized as follows:
Section \ref{SEC_QMLE} introduces the QML procedure for the general class of conditionally heteroscedastic time series models of multiplicative form and investigates the asymptotic properties of the Quasi-Likelihood (QL) function (Section \ref{SEC_AC}). 
Next, in Section \ref{sec:qml_garchpq}, we present the QML estimation of the GARCH parameters.
In Section \ref{sec:adavol}, we present our adaptive approach for recursively estimating GARCH parameters named AdaVol.
We examine the AdaVol procedure on simulated and real-life observations in Section \ref{sec:applications}, and some concluding remarks are made in Section \ref{sec:discussion}.


\section{QML Estimation in Conditionally Heteroscedastic Time Series Models} \label{SEC_QMLE}


The approximate QMLE $\widehat{\theta}_n^{*}$ is defined as
\begin{align} \label{GARCH_QLIK}
\widehat{\theta}_n^{*} \in \argmin_{\theta \in \K} \widehat{L}_n(\theta),
\end{align}
where the parameter set $\K$ is a suitable compact subset of the parameter space $\Theta$. 
The QL function $L_n (\theta)$ and approximate QL function $\widehat{L}_n (\theta)$ are given by
\begin{align} \label{QMLE_DEF} 
L_n (\theta) = \sum_{t=1}^n l_t (\theta) 
\text{ and }
\widehat{L}_n (\theta) =\sum_{t=1}^n \widehat{l}_t (\theta),
\end{align}
with QL losses, denoted $l_t (\theta)$ and $\widehat{l}_t (\theta)$, given as
\begin{align} \label{ONE_LIKE_HAT}
l_t (\theta) =  \frac{1}{2}  \left(\frac{X_t^2}{h_t(\theta)} + \log h_t(\theta) \right)
\text{ and }
\widehat{l}_t (\theta) =  \frac{1}{2} \left(\frac{X_t^2}{\widehat{h}_t(\theta)} + \log \widehat{h}_t(\theta) \right),
\end{align}
where $(\widehat{h}_t)$ is an approximation of $(h_t)$ defined recursively for $t \geq 1$ as in \eqref{CHTSM_h} with initialization $\widehat{h}_{-q+1} = \cdots = \widehat{h}_{0} = 0$ or any deterministic constant.
From \cite[Proposition 5.2.12]{straumann2005}, we know the initialization error between $(\widehat{h}_t)$ and the true $(h_t)$ will vanish exponentially fast almost surely.
Assuming $Z_0$ is standard normal distributed, we may note $X_t$ is also Gaussian with variance $h_t$ conditioned on $\F_{t-1}$.
The QL function $L_n (\cdot)$ in \eqref{QMLE_DEF} is derived under this Gaussian assumption.

The consistency and asymptotic properties of the QMLE $\widehat{\theta}_n^{*}$ combined with the robustness of the QL function for overdispersion make the method highly used in practice (e.g., see \cite{patton2006}). 
Under the assumptions in \cite[N.1, N.2, N.3 and N.4]{straumannmikosch2007}, the QMLE $\widehat{\theta}_n^{*}$ is strongly consistent and asymptotically normal, that is
\begin{align} \label{eq:qmle_asymp_prop}
\widehat{\theta}_n^{*} \overset{\text{a.s.}}{\rightarrow} \theta_0
 \text{ and } 
\sqrt{n} \left( \widehat{\theta}_n^{*} - \theta_0 \right) \rightarrow \NORM \left(0, V_0 \right)
\text{ as } n \rightarrow \infty,
\end{align}
with $\theta_0$ as the true parameter vector and $V_0$ the asymptotic covariance matrix.

Unfortunately, these asymptotic properties in \eqref{eq:qmle_asymp_prop} come with a drawback on the QL loss; the consistency is achieved through careful domination of logarithmic moments. 
The concavity of logarithms makes the criterion insensitive to extreme values, but it also implies that the criterion itself behaves as a concave function. 
As most optimization algorithms are based on convex assumptions, this is striking.

In the next section, we show that the approximate Hessian $\widehat{H}_n(\theta) = n^{-1}\nabla_{\theta}^2 \widehat{L}_{n} \left(\theta \right)$ admits strictly positive eigenvalues for $n$ sufficiently large dependent on the model specifications and the underlying data process. 
This means that for sufficiently large batch sizes of observations, the QMLE $\widehat{\theta}_n^{*}$ can be seen as the unique solution of a locally strongly convex optimization problem; the existence and uniqueness of $\widehat{\theta}_n^{*}$ ensure that usual iterative optimization routines can efficiently approximate it for $n$ large enough.


\subsection{Asymptotic Properties of the QL Function} \label{SEC_AC}


To establish the asymptotic local convexity of the QL function of the model described by \eqref{CHTSM}-\eqref{CHTSM_h}, we need the following assumptions: Assumption \ref{assump:cond_exists_solution}, \ref{assump:cond_bounds}, and \ref{assump:cond_comp_inde}, which naturally emerges from the arguments and properties \cite{straumannmikosch2007} made to ensure stability of the QL function and QMLE procedure. 
We will use two different matrix norms, namely, let $\norm{A}_{op}$ denote the matrix operator norm of the matrix $A \in \R^{d \times d}$ with respect to the Euclidean norm, i.e., $\norm{A}_{op} = \sup_{v \neq 0} \vert Av \vert/\vert v \vert$, and denote $\norm{A}_{\K}$ the norm of the continuous matrix-valued function $A$ on $\K$, i.e., $\norm{A}_{\K} = \sup_{x \in \K} \norm{A(x)}_{op}$, where $\K$ is a compact set of $\R^d$.

\begin{assumpW}\label{assump:cond_exists_solution}
The model \eqref{CHTSM}-\eqref{CHTSM_h} with $\theta = \theta_0$ admits a unique stationary ergodic solution.
\end{assumpW}

\begin{assumpW} \label{assump:cond_bounds}
Let $\K \subset \Theta$ be a compact set with true parameter vector $\theta_0 \in \K$ in the interior. 
The random functions fulfill certain conditions, such that $\E [\norm{l_0}_\K] < \infty$, $\E [\norm{ \nabla_{\theta}^{2} l_0}_\K] < \infty$, and furthermore have the following uniform convergences 
$\norm{n^{-1} \widehat L_n  - L_n}_\K \overset{\text{a.s.}}{\longrightarrow} 0$
and 
$n^{-1}\norm{\nabla_{\theta}^{2} \widehat{L}_n  - \nabla_{\theta}^{2} L_n }_\K \overset{\text{a.s.}}{\longrightarrow} 0$
as $n \rightarrow \infty$.
\end{assumpW}

\begin{assumpW}\label{assump:cond_comp_inde}
The components of the vector $\nabla_{\theta} g_{\theta}(X_0,h_0)$ from \eqref{CHTSM_h} with $\theta = \theta_0$ are linearly independent random variables.
\end{assumpW}

The following Theorem \ref{THM_Convex} is an extension of \cite{ip2006}, which established similar results for the likelihood function of GARCH models under the assumption that $(X_t)$ is strictly stationary and strongly mixing with geometric rate, and $(Z_t)$ is Gaussian. 
Solving the QML estimation problem in \eqref{GARCH_QLIK} for $\widehat{\theta}_n^{*}$ is known to be computationally heavy as one has to find the solution of a non-linear equation, namely \eqref{QMLE_DEF}. 
Nonetheless, Theorem \ref{THM_Convex} ensures the existence of an $N$ such that we have a unique global QMLE $\widehat{\theta}_n^{*}$ for all $n \geq N$.

\begin{thm} \label{THM_Convex}
Under Assumption \ref{assump:cond_exists_solution}, \ref{assump:cond_bounds}, and \ref{assump:cond_comp_inde}, there exist positive constants $C, \delta>0$, and a random positive integer $N \in \N$ such that
\begin{align} \label{CONVEX_HESSIAN_QLIK}
g^T\widehat{H}_n (\theta ) g > C g^Tg, \text{ } \forall n \geq N, \quad \text{a.s.,}
\end{align}
for all $\theta \in B(\theta_0,\delta)$ and $g \in \R^{d} \setminus \{0\}$.
\end{thm}

The result above shows local strong convexity of the QL function $\widehat{L}_n$. 
The following corollary arises from the proof of Theorem \ref{THM_Convex}:

\begin{cor} \label{COR_Convex}
Under Assumption \ref{assump:cond_exists_solution}, \ref{assump:cond_bounds}, and \ref{assump:cond_comp_inde}, the QMLE $\widehat{\theta}_n^{*}$ exists and is unique, that is
\begin{align*}
\widehat{\theta}_n^{*} = \argmin_{\theta \in \K} \widehat{L}_n (\theta).
\end{align*}
\end{cor}

Local strong convexity is crucial for guaranteeing the convergence of an optimization algorithm, although some methods go beyond this point (\cite{ward2018}).
Thus, Theorem \ref{THM_Convex} is an essential result for computing the QMLE $\widehat{\theta}_n^{*}$ parameters of the model in \eqref{CHTSM}-\eqref{CHTSM_h}. 
Nevertheless, to guarantee the property in \eqref{CONVEX_HESSIAN_QLIK}, we need a sufficiently large (and maybe unbounded) random $N$, which depends on the true parameter vector $\theta_{0}$, the parameter estimates $(\widehat{\theta}_{t}^{*})$, and the observations $(X_t)$. 
One often has a fixed size of observations in practice, so the iterative algorithm may not converge. 
To our experience, this phenomenon may occur when the true parameter vector $\theta_{0}$ is close to the boundary of $\K$, or if the initial values $\widehat{\theta}_{0}^{*}$ are far away from the true parameters $\theta_{0}$.


\subsection{QML Estimation of GARCH$(p,q)$ Parameters} \label{sec:qml_garchpq}


The general class of conditionally heteroscedastic time series models includes the very popular ARCH and GARCH models. 
For more than three decades, these models have attracted considerable amounts of attention in the literature since their introduction.
A process $(X_t)$  is called a GARCH$(p,q)$ process with parameter vector $ \theta = ( \omega , \alpha_1, \dots , \alpha_p, \beta_1 , \dots, \beta_q )^T$, if it satisfies  
\begin{align} \label{EQ_GARCH}
    \begin{cases}
      X_t = \sigma_t Z_t,\\
       \sigma^2_t = \omega + \sum_{i=1}^p \alpha_i  X_{t-i}^2 + \sum_{j=1}^q \beta_j \sigma^2_{t-j}, 
    \end{cases}
\end{align}
where $\omega$, $\alpha_i$, and $\beta_j$ for $1 \leq i \leq p$ and $1 \leq j \leq q$ are non-negative parameters ensuring the non-negativity of the conditional variance process $(\sigma_t^2)$. 
The innovations $(Z_t)$ is a sequence of i.i.d. random variables with $\E [Z_0] = 0$ and $\E [Z_0^2]=1$. 
Likewise, one can define an ARCH$(p)$ process by setting $\beta_j=0$ for $1 \leq j \leq q$ in \eqref{EQ_GARCH}. 
The GARCH$(p,q)$ process $(X_t)$ given in \eqref{EQ_GARCH} has QL losses given by $\widehat{l}_t(\theta) = 2^{-1} (X_t^2/\widehat{\sigma}_t^2 (\theta) + \log \widehat{\sigma}_t^2 (\theta))$ with first-order derivative
\begin{align} \label{GARCH_QLIK_D}
 \nabla_{\theta} \widehat{l}_t (\theta)
= \nabla_{\theta} \widehat{\sigma}_t^2 (\theta) \left( \frac{ \widehat{\sigma}_t^2(\theta) - X_t^2}{2 \widehat{\sigma}_t^4(\theta)} \right)
\end{align}
and second-order derivative
\begin{align} 
 \nabla_{\theta}^{2} \widehat{l}_t (\theta)
=
  \nabla_{\theta} \widehat{\sigma}_t^2 (\theta)^T  \nabla_{\theta} \widehat{\sigma}_t^2 (\theta) \left(\frac{2X_t^2 - \widehat{\sigma}_t^2(\theta) }{2 \widehat{\sigma}_t^6(\theta)} \right) 
+ \nabla_{\theta}^{2} \widehat{\sigma}_t^2 (\theta) \left( \frac{ \widehat{\sigma}_t^2(\theta) - X_t^2}{2 \widehat{\sigma}_t^4(\theta)} \right), \label{GARCH_QLIK_DD}
\end{align}
where $\nabla_{\theta} \widehat{\sigma}_t^2 (\theta) = \vartheta_t(\theta) + \sum_{j=1}^q   \beta_j  \nabla_{\theta} \widehat{\sigma}_{t-j}^2 (\theta)$ with $\vartheta_t (\theta)= (1, X_{t-1}^2, \dots, X_{t-p}^2, \widehat{\sigma}_{t-1}^2(\theta), \dots, \widehat{\sigma}_{t-q}^2 (\theta)  )^T \in \R^{p+q+1}$ and Hessian $\widehat{H}_n (\theta)
= n^{-1} \sum_{t=1}^{n} \nabla_{\theta}^{2} \widehat{l}_t (\theta)$.

The equations \eqref{EQ_GARCH} creates a complicated probabilistic structure that is not easily understood, although it looks relatively simple.
The conditions ensuring the existence and uniqueness of a stationary solution to the equations \eqref{EQ_GARCH} for GARCH$(1,1)$ was provided by \cite{nelson1990}.
\cite{bougerolpicard1992} later showed it for the GARCH$(p,q)$ model using that GARCH$(p,q)$ can be embedded in a Iterated Random Lipschitz Map (IRLM). 
See \cite{bougerol1993} for a formal definition of IRLMs.

We can illustrate the IRLM method on the GARCH$(1,1)$ model with parameter vector $\theta = (\omega,\alpha_1,\beta_1)^T$. 
The IRLM for $\sigma_t^2$ is then given by $\sigma_t^2 = A_t \sigma_{t-1}^2 + B_t$ with $t \in \Z$, where $A_t = \alpha_1 Z_{t-1}^2 + \beta_1$ and $B_t = \omega$. 
Note $\left( (A_t,B_t) \right)$ constitutes an i.i.d. sequence.
From the literature on IRLMs it is well known that the conditions $ \E [\log\vert A_0 \vert] < 0$ and $\E [\log^+ \vert B_0 \vert ]< \infty$ guarantee the existence and uniqueness of a strictly stationary solution of the IRLM $Y_t = A_t Y_{t-1} + B_t$ for $t \in \Z$ provided $\left( (A_t,B_t) \right)$ is a stationary ergodic sequence. 
Applying this to the GARCH$(1,1)$ model, we get the known sufficient condition for the existence of a stationary solution, namely $\E [ \log(\alpha_1 Z_0^2 + \beta_1)] <0$.
This also implies $\beta_1 < 1$ since $\log(\beta_1) \leq \E [ \log(\alpha_1 Z_0^2 + \beta_1)] < 0$. 
Likewise, the ARCH$(1)$ process ($\beta_1 = 0$) then requires $\E [ \log ( \alpha_1 Z_0^2 ) ] < 0$, which is the same as $\alpha < 2e^\epsilon \approx 3.56$ with $Z_0$ being Gaussian.
Thus, the stationary condition is much weaker than the second-order stationary condition in which we require $\alpha_1 + \beta_1 < 1$.

The statistical inference leads to further nontrivial problems since the exact distribution of $(Z_t)$ remains unspecified, and so one usually determines the likelihoods under the hypothesis of standard Gaussian innovations. 
Moreover, the volatility $\left(\sigma_t\right)$ is an unobserved quantity approximated by mimicking the recursion \eqref{EQ_GARCH} with an initialization, for instance $X_{-p+1} = \cdots = X_0 = 0$ and $\sigma_{-q+1}^2 = \cdots = \sigma_0^2 = 0$. 
\cite{berkes2003} showed under minimal assumptions that the QMLE is strongly consistent and asymptotically normal.

Furthermore, under Assumption \ref{assump:cond_exists_solution}-\ref{assump:cond_comp_inde}, we have asymptotic local strong convexity of the QL function in GARCH$(p,q)$ models by Theorem \ref{THM_Convex}. 
However, the number of observations needed to guarantee local strong convexity vary. 
This can easily be seen by looking at the simplest case, namely when $(X_t)$ follows an ARCH$(1)$ process with parameter vector $\theta = (\omega, \alpha_1)^T$. 
The volatility process $\sigma^2_t (\theta)$ is given as $\omega+ \alpha_1 X_{t-1}^2$.
The eigenvalues of $\nabla_{\theta}^{2} l_t (\theta)$ are given by $\lambda_{t} = (\lambda_{t,1}, \lambda_{t,2}) = (0, \lambda_{t,2})$ with $\lambda_{t,2} = (1+X_{t-1}^4)(2 X_t^2 - \sigma_t^2(\theta)) 2^{-1}\sigma_{t}^{-6}(\theta)$.
Thus, the non-negativity of $\lambda_{t,2}$ would ensure convexity at time $t$ in our QML procedure. 
However, the probability of having convexity at each $t$ is unlikely as $\Prob ( \cap_{t=1}^n \nabla_{\theta}^{2} l_t (\theta) \geq 0 ) = \Prob ( \cap_{t=1}^n Z^2_t \geq 1/2) = \Prob ( Z^2_0 \geq 1/2 )^n$ is approximately $0.52^n$ with i.i.d. Gaussian innovations $(Z_t)$, i.e., $(Z_{t}^{2})$ is $\chi^{2}$-distributed with $1$ degree of freedom.
On the other hand, increasing the number of observations used at each iteration would increase the probability of having local strong convexity.


\section{Adaptive Recursive QML Estimation} \label{sec:adavol}


Our recursive QML method relies on stochastic approximations introduced by \cite{robbins1951}, which only requires the previous parameter estimate to update the parameter estimate using the new observation. 
We perform the first-order stochastic gradient method defined as
\begin{align} \label{eq:sgd}
\widehat{\theta}_{t} = \widehat{\theta}_{t-1} - \eta_{t-1} \nabla_{\theta} \widehat{l}_t (\widehat{\theta}_{t-1}),
\end{align}
where $\eta_{t-1}>0$ is the step-size at the $t-1$ step, and $\nabla_{\theta} \widehat{l}_t (\widehat{\theta}_{t-1})$ is the gradient using the $X_t$ observation and the QMLE estimate $\widehat{\theta}_{t-1}$.
This method is computationally efficient as it only requires a cost of $\bigO(d)$ per recursion.
Depending on the number of observations, we have a trade-off between the accuracy of the recursive QML estimates and the time it takes to perform a parameter update (\cite{bottou2007}).

According to \cite{robbins1951}, we must schedule the step-size such that $\sum_{t=1}^{\infty} \eta_{t}=\infty$ and $\sum_{t=1}^{\infty} \eta_{t}^{2} < \infty$, but these bounds do not make the choice of an appropriate step-size $\eta_{t}$ easier in practice. 
A more suitable approach is an adaptive learning rate, which updates the step-size in \eqref{eq:sgd} on the fly pursuant to the gradient $\nabla_{\theta} \widehat{l}_t(\cdot)$. 
Thus, our choice of step-size $\eta_{t}$ have less impact on performance, making convergence more robust and lower the demand for manually fine-tuning.
Such an approach is often used in settings of streaming data as generic methods are preferred.
Adaptive and separate learning rates for each parameter was proposed by \cite{duchi2011} in their AdaGrad procedure. 
A different learning rate speeds up convergence in situations where the appropriate learning rates vary across parameters. 
Other well-known examples of adaptive learning rates could be AdaDelta by \cite{zeiler2012}, RMSProp by \cite{tieleman2012} and ADAM by \cite{kingma2015}. 
As we may expect a lack of convexity, we select the AdaGrad algorithm since it has shown promising results in non-convex optimization (\cite{ward2018}). 
The AdaGrad procedure is given by the updates
\begin{align}  \label{eq:sgd_adagrad}
\widehat{\theta}_{t} = \widehat{\theta}_{t-1} - \frac{\eta}{\sqrt{ \sum_{i=1}^{t} \nabla_{\theta} \widehat{l}_{i} (\widehat{\theta}_{i-1})^2 + \eps}}  \nabla_{\theta} \widehat{l}_t (\widehat{\theta}_{t-1}),
\end{align}
where $\eta>0$ is a constant learning rate and $\eps>0$ a small number ensuring positivity.
Good default values are $\eta = 0.1$ and $\eps = 10^{-8}$, see the AdaVol algorithm in Table \ref{algo:online_garchpq}.
Note $\nabla_{\theta} \widehat{l}_{i} (\widehat{\theta}_{i-1})^2$ denotes the element-wise square $\nabla_{\theta} \widehat{l}_{i} (\widehat{\theta}_{i-1}) \odot \nabla_{\theta} \widehat{l}_{i} (\widehat{\theta}_{i-1})$.

As the QL loss is defined only for $\widehat{\theta}_n \in \K$, we will require that the recursive algorithm always takes values in $\K$. 
\cite{zinkevich2003} suggests we project our approximation $\widehat{\theta}_n$ onto $\K$, preventing large jumps and enforcing the convergence of our stochastic gradient method. 
By implementing this projection on \eqref{eq:sgd_adagrad}, we have our method for updating estimates, namely
\begin{align} \label{eq:sgd_adagrad_proj}
\widehat{\theta}_{t} = \text{P}_{\K} \left[ \widehat{\theta}_{t-1} -  \frac{\eta}{\sqrt{\sum_{i=1}^{t} \nabla_{\theta} \widehat{l}_{i} (\widehat{\theta}_{i-1})^2   + \eps}} \nabla_{\theta} \widehat{l}_t (\widehat{\theta}_{t-1}) \right].
\end{align}


\subsection{Adaptive Recursive QML Estimation for GARCH Models} \label{sec:rec_qml_garch}


The GARCH process $(X_t)$ parameters can be numerically challenging to estimate in empirical applications. 
The numerical optimization algorithms can quickly fail or converge to irregular solutions (\cite{zumbach2000pitfalls}). 
Therefore, examining the approximative QMLE $\widehat{\theta}_n^*$ must be made with a healthy amount of skepticism.
A well-discussed problem for the GARCH$(p,q)$ models is that the QMLE performs poorly for numerically small (but still positive) values of $\omega$.
The parameter $\omega$ is vital and often tricky to estimate.
Stabilizing the estimation of $\omega$ would not only improve the $\omega$ estimate but also have a positive impact on the other model parameters.

On way to overcome small values of $\omega$ for the GARCH$(p,q)$ model is by scaling $(X_t)$ with some factor $\lambda > 0$ as we have homogeneity; let $(X_t)$ follow a GARCH$(p, q)$ process with parameter vector $\theta = (\omega, \alpha_1, \dots, \alpha_p, \beta_1, \dots, \beta_q)^T$ and innovations $(Z_t)$. Then for any $\lambda > 0$, the process $(\sqrt{\lambda}X_t)$  is a GARCH$(p,q)$ process with parameter vector $\theta = (\lambda\omega, \alpha_1, \dots, \alpha_p, \beta_1, \dots, \beta_q)^T$ and identical innovations $(Z_t)$.

However, we wish to avoid this form of inference in our recursive algorithm as one then needs to come up with a scaling parameter that has to be estimated beforehand. 
Instead, we circumvent this issue by introducing a concept called Variance Targeting Estimation (VTE) (\cite{francq2011}). 
We apply VTE for estimating $\omega$ by use of $\gamma^2$, which is the unconditional variance estimated by the sample variance (as seen in \eqref{eq:omega_reparam}). 
Thus we have a two-step estimator where we estimate the sample variance $\gamma^2$ recursively, and the remaining parameters $\theta = (\alpha_1, \dots , \alpha_p, \beta_1 , \dots, \beta_q )^T$ are estimated by the QML method. 
Pseudo-code of the AdaVol algorithm is presented in Table \ref{algo:online_garchpq}.
The reparametrization is obtained by defining
\begin{align} \label{eq:omega_reparam}
\omega = \gamma^2 \left( 1 -  \sum_{i=1}^p \alpha_i  - \sum_{j=1}^q \beta_j  \right).
\end{align}
The volatility process in the GARCH$(p,q)$ process can then be rewritten as
\begin{align} \label{eq:garch_reparam}
(\sigma^2_t - \gamma^2) = \sum_{i=1}^p \alpha_i  (X_{t-i}^2-\gamma^2) + \sum_{j=1}^q \beta_j (\sigma^2_{t-j}-\gamma^2).
\end{align}
Similarly, one can define an ARCH$(p)$ process by setting $\beta_j=0$ for $1 \leq j \leq q$.
The GARCH$(p,q)$ process $(X_t)$ in \eqref{eq:garch_reparam} has similar QL losses as before except $\nabla_{\theta} \widehat{\sigma}_t^2 (\theta)$ in \eqref{GARCH_QLIK_D} and \eqref{GARCH_QLIK_DD}, where $\vartheta_t (\theta)$ is given as $( X_{t-1}^2 - \gamma^2, \dots, X_{t-p}^2 - \gamma^2, \widehat{\sigma}_{t-1}^2(\theta)- \gamma^2, \dots, \widehat{\sigma}_{t-q}^2 (\theta) - \gamma^2 )^T \in \R^{p+q}$ and the parameter space is defined by $\K= \left\{(\alpha_1,\ldots,\alpha_p,\beta_1,\ldots,\beta_q)\in\R_+^{p+q} \middle| \sum_{i=1}^p \alpha_i  + \sum_{j=1}^q \beta_j < 1 \right\}$.

\begin{table}[h!]
\removelatexerror
\begin{algorithm}[H]
\SetAlgoRefName{}
\SetAlgoLined
\DontPrintSemicolon
\KwData{
$(X_t)_{t \geq 1}$ (observations) \;
}
\SetKwInOut{Input}{input}\SetKwInOut{Output}{output}
\Input{
$\widehat{\theta}_0$ (initial parameter vector), $\eta=0.1$, $\eps = 10^{-8}$ \;
}
\Begin{
initialize: $\widehat{\sigma}_{1}^{2} = X_{1}^{2}$, $\widehat{\mu}_{0} = 0$, $\widehat{\gamma}_{0}^{2} = 0$, $\widehat{G}_{0} = \eps$ and $t=0$ \;

\While{$\widehat{\theta}_{t}$ \KwSty{not converged}}{
$t = t+1$ \;
$\widehat{\mu}_{t} = t(t+1)^{-1} \widehat{\mu}_{t-1} + (t+1)^{-1} X_t$ 
\;
$\widehat{\gamma}^2_{t} = (t-1)t^{-1} \widehat{\gamma}^2_{t-1} + t^{-1} \left(X_t - \widehat{\mu}_{t}\right)^{2}$ 
\;
$\widehat{g}_{t} = \nabla_{\theta} \widehat{l}_{t} (\widehat{\theta}_{t-1}) $ 
\;
$\widehat{G}_{t} = \widehat{G}_{t-1} + \widehat{g}_{t}^{2}$ 
\;
$\widehat{\theta}_{t} = \text{P}_{\K} \left[ \widehat{\theta}_{t-1} - \eta \widehat{G}_{t}^{-1/2} \widehat{g}_{t} \right]$ 
\;
$\widehat{\sigma}_{t+1}^2 = \widehat{\gamma}^2_{t} +\sum_{i=1}^{p} \widehat{\alpha}_{i}^{(t)} (X_{t-i}^{2} - \widehat{\gamma}^2_{t}) + \sum_{j=1}^{q} \widehat{\beta}_{j}^{(t)} (\widehat{\sigma}_{t-j}^2-\widehat{\gamma}^2_{t})$
}}
\KwResult{$\widehat{\theta}_{t}$ (resulting estimates), $\widehat{\sigma}_{t+1}^{2}$ (predicted volatility)}
\SetAlgorithmName{AdaVol}{}{}
\caption{Adaptive recursive QML estimation for GARCH$(p,q)$ models using the technique of VTE.}
\end{algorithm}
\caption{Pseudo-code of the AdaVol algorithm.}
\label{algo:online_garchpq}
\end{table}

The VTE is not a requirement for the recursive method, but it provides additional speed and numerical stability.
Namely, the VTE ensures a consistent estimate of the long-run variance, even if the model is misspecified. 
Additionally, presuming $\gamma$ is well estimated, we reduce the parameter space dimension and increase the speed of convergence of the recursive optimization routines. 
Moreover, the geometry of the new set of optimization $\K$ allows the projection step in \eqref{eq:sgd_adagrad_proj} to be efficiently implemented following \cite{duchi2008}.

One should be aware that the VTE requires stronger assumptions for the existence of the variance and is likely to suffer from efficiency loss. 
\cite{francq2011} also showed that the VTE would never be asymptotically more accurate than the QMLE. 
Another drawback of using the VTE is the need for a finite fourth moment of the process $(X_t)$. 
Meaning, one would need $\alpha_1<0.57$ for an ARCH$(1)$ model using standard Gaussian noise as $EX_t^4 < \infty$ if and only if $\alpha_1^2 + (EZ_0^4-1)\alpha_1^2 < 1$. 
For a GARCH$(1,1)$ model, we should have $(\alpha_1+\beta_1)^2 + (EZ_0^4-1)\alpha_1^2 < 1$. 
These parameter bounds restrict the usefulness and range of applications for the VTE techniques.
Fortunately, these constraints solely concern the batch setting.


\section{Applications} \label{sec:applications}


In this section, we examine the AdaVol algorithm on simulated and real-life observations. 
Our implementation of AdaVol is provided in a repository at \cite{werge2019}, and a relative speed comparison can be found in \ref{sec:appendix:speed}.
We compare our approach to the Iterative QMLE (IQMLE) approximation $\widetilde{\theta}_n$, which is estimated at every two thousand increments using all observations up to this point, i.e., $(\widetilde{\theta}_t)_{(k-2000)+1 \leq t \leq k}$ is estimated using $(X_t)_{1 \leq t \leq k}$ for $k=2000,4000, \dots, n$.
In this way, we illuminate the large-scale learning trade-off of applying our recursive method instead of the iterative method, which is forward-looking with up to two thousand observations (\cite{bottou2007}).
As suggested by \cite{ip2006}, we use the (bounded) \emph{L-BFGS} algorithm to solve the nonlinear optimization problem in \eqref{GARCH_QLIK} for $\widetilde{\theta}_n$ with initial guess $\widetilde{\theta}_0 \in \K$.
Our recursive QMLE approximation $\widehat{\theta}_n$ is produced by the AdaVol algorithm (described in Table \ref{algo:online_garchpq}).
It takes our initial value $\widehat{\theta}_0 \in \K$, learning rate $\eta=0.1$ and $\epsilon=10^{-8}$ as input. 
At last, for a fair comparison, we always use the same initial guess for both methods, namely $\widehat{\theta}_0 = \widetilde{\theta}_0 \in \K$.

It is possible to customize AdaVol by tuning the learning parameter $\eta$, e.g., by choosing the best performing learning rate evaluated on the first part of the observations. 
We use a fixed learning rate $\eta=0.1$ across all applications (simulated and real-life observations) to avoid the learning rate's potential influence in our experiments.
However, one should be aware of the versatility achieved with different learning rate choices.
The choice of learning rates is cumbersome, as an excessive learning rate can cause the algorithm to deviate from the true parameter estimate.
In contrast, a learning rate that is too small can lead to slow convergence.
Nevertheless, a small learning rate may be preferred if one only wants to keep track of minor parameter estimation changes.


\subsection{Simulations} \label{sec:simulations}


All simulations are performed by the use of twenty thousand observations $(n=20000)$, and the simulated data $(X_t)$ is always generated using Gaussian innovations with zero mean and unit variance.
To avoid possible bias due to the choice of the true parameter vector $\theta_0$ and initial values $\widehat{\theta}_0,\widetilde{\theta}_0$, we conduct our experiments using random parameter vectors $\theta_0 \in \K$ and random initial guesses $\widehat{\theta}_0,\widetilde{\theta}_0 \in \K$.
These parameter vectors are drawn randomly from our parameter space $\K$.
The $\omega$ parameter is generated by taking a positive number from a uniform distribution, and then we multiply it with $10^{-\tau}$, where $\tau$ is some random positive integer up to eight.
In this way, we cover a broad parameter domain while having parameter values close to the boundary.
Similarly, the $(\alpha_i)_{1 \leq i \leq p}$ and $(\beta_j)_{1 \leq j \leq q}$ parameters is generated from a uniform distribution with the condition of having $\sum_{i=1}^p \alpha_i+\sum_{j=1}^{q} \beta_j < 1$. 
Note that the initial guesses $\widehat{\theta}_0$ and $\widetilde{\theta}_0$ are generated the same way.
Thus, when we mention random parameters for the rest of the paper, we refer to this generation procedure.


\subsubsection{ARCH Models}


As discussed earlier, the iterative QMLE approximation $\widetilde{\theta}_n$ performs poorly for numerically small $\omega>0$ values, which are often encountered in financial time series.
Before moving on to the case of small $\omega$ parameter values, we have in Figure \ref{fig:plt_arch_trajectory_large_omega} the trajectories of both QMLE approximations using an ARCH$(1)$ process with true parameter vector and initial values given by
\begin{align} \label{sim_arch_theta0_omega}
\theta_0 =
 \begin{pmatrix} \omega \\ \alpha_1 \end{pmatrix}
 =
 \begin{pmatrix} 2.0 \\ 0.6 \end{pmatrix}
\text{ and }
\widehat{\theta}_0 
= \widetilde{\theta}_0 
= \begin{pmatrix} 1.5 \\ 0.4 \end{pmatrix}.
\end{align}

\begin{figure}[h!]
\centering
\includegraphics[width=1.0\linewidth]{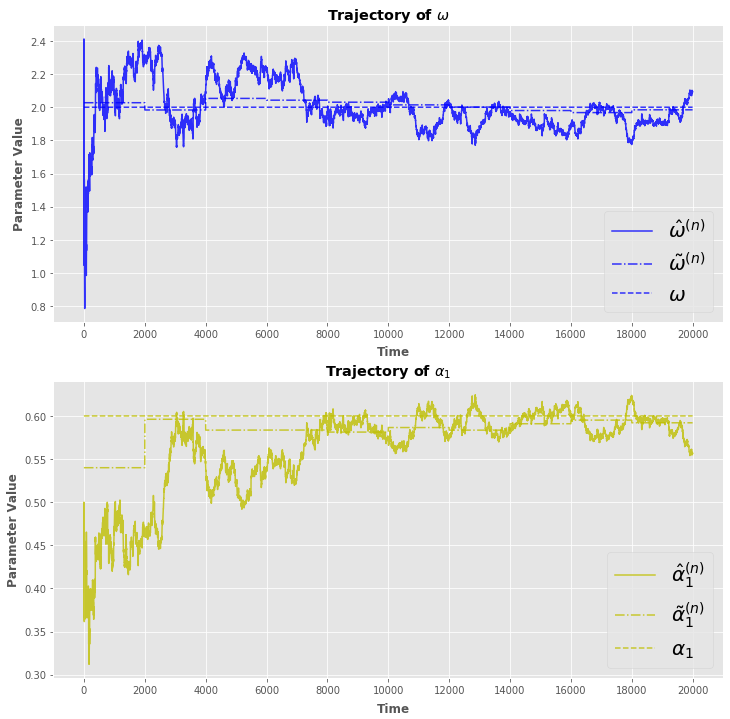}
\caption{Trajectory of $\widehat{\theta}_n$ (solid line) and $\widetilde{\theta}_n$ (semi-dotted line) for an ARCH$(1)$ process with true parameter vector (dotted line) and initial guess given in \eqref{sim_arch_theta0_omega}.}
\label{fig:plt_arch_trajectory_large_omega}
\end{figure}

Figure \ref{fig:plt_arch_trajectory_large_omega} shows a very reasonable convergence of both estimators, $\widehat{\theta}_n = (\widehat{\omega}^{(n)}, \widehat{\alpha}_{1}^{(n)})^T$ and $\widetilde{\theta}_n = (\widetilde{\omega}^{(n)}, \widetilde{\alpha}_{1}^{(n)})^T$, when the true parameter $\omega = 2.0$. 
Not surprisingly, our method experiences some fluctuations initially, but as the learning rate decreases, the fluctuation likewise evaporates, and within the first few thousand observations, we hit the true parameter values.

Likewise, in Figure \ref{fig:plt_arch_trajectory}, we have the QMLE approximations' trajectories for an ARCH$(1)$ process, but now with true parameter vector and initial guess given as
\begin{align} \label{sim_arch_theta0_1}
\theta_0 =
 \begin{pmatrix} 1 \cdot 10^{-8} \\ 0.6 \end{pmatrix}
 \text{ and }
\widehat{\theta}_0 
= \widetilde{\theta}_0 
= \begin{pmatrix} 5 \cdot 10^{-8} \\ 0.4 \end{pmatrix}.
\end{align}
Figure \ref{fig:plt_arch_trajectory} indicates a modest convergence of $\widehat{\theta}_n$ but shows slow convergence of $\widetilde{\alpha}_n$ towards the true $\alpha_1$ parameter.
In addition, $\widetilde{\alpha}_n$ seems biased concerning the initial value $\widetilde{\alpha}_0 = 0.4$ as it processes almost half of the observations before moving closer to the true $\alpha_1 = 0.6$.

\begin{figure}[h!]
\centering
\includegraphics[width=1.0\linewidth]{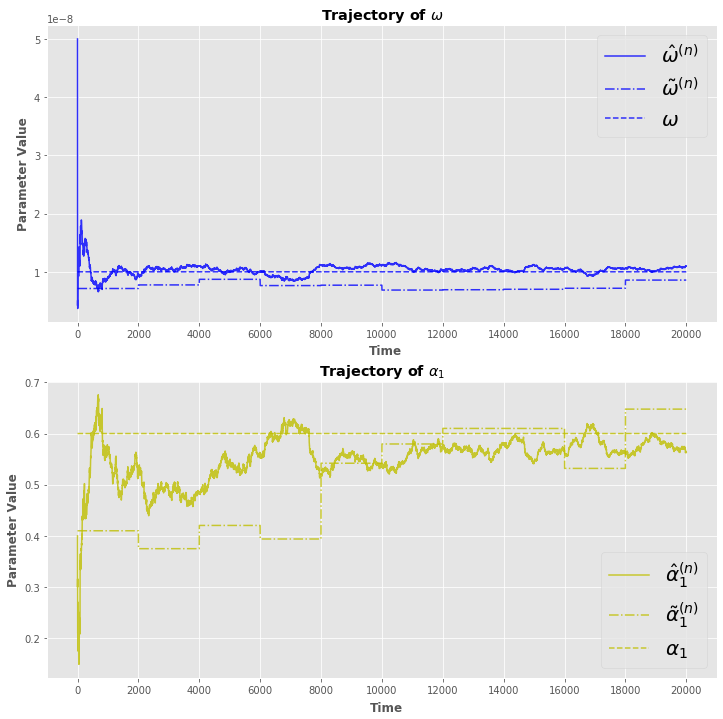}
\caption{Trajectory of $\widehat{\theta}_n$ (solid line) and $\widetilde{\theta}_n$ (semi-dotted line) for an ARCH$(1)$ process with true parameter vector (dotted line) and initial guess given in \eqref{sim_arch_theta0_1}.}
\label{fig:plt_arch_trajectory}
\end{figure}

A way of demonstrating the variation of $\widehat{\theta}_n$ and $\widetilde{\theta}_n$ performance for small $\omega$ values is presented in Figure \ref{fig:plt_arch_trajectory_mc_online} and Figure \ref{fig:plt_arch_trajectory_mc_offline}, where we have the average trajectory of one hundred trajectories with their corresponding boxplots showing the distribution of these one hundred trajectories.

\begin{figure}[h!]
\centering
\includegraphics[width=1.0\linewidth]{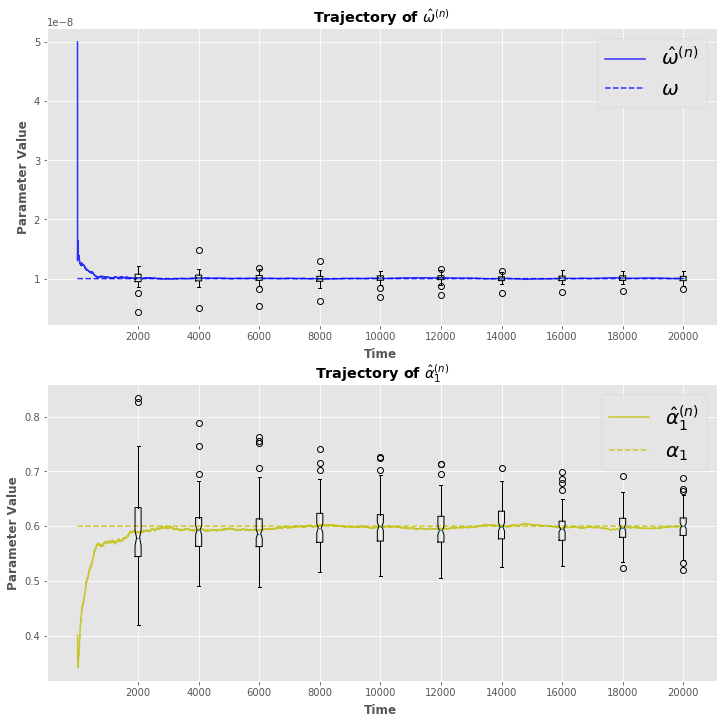}
\caption{Average trajectory (solid line) of one hundred $\widehat{\theta}_n$'s for an ARCH$(1)$ process with true parameter vector (dotted line) and initial guess from \eqref{sim_arch_theta0_1}. The boxplots shows the distribution of the one hundred trajectories.}
\label{fig:plt_arch_trajectory_mc_online}
\end{figure}

\begin{figure}[h!]
\centering
\includegraphics[width=1.0\linewidth]{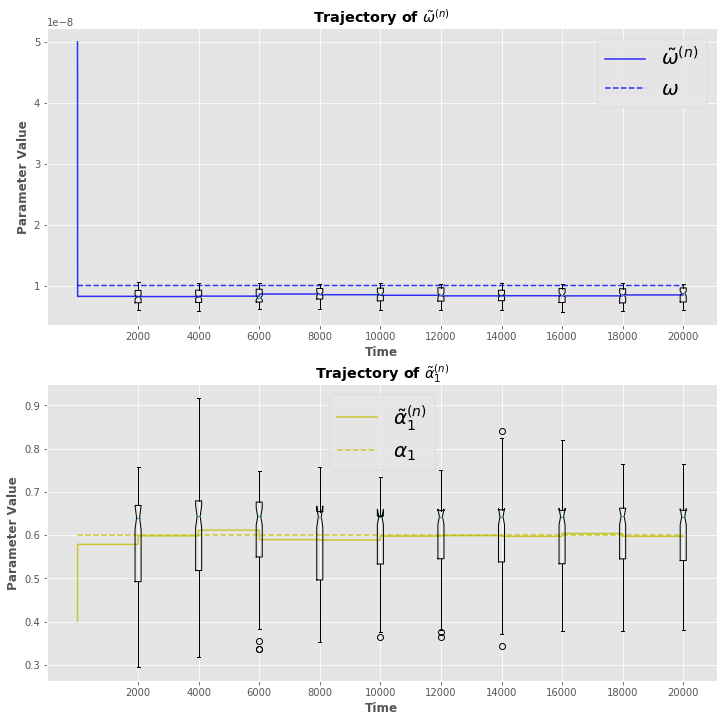}
\caption{Average trajectory (solid line) of one hundred $\widetilde{\theta}_n$'s for an ARCH$(1)$ process with true parameter vector (dotted line) and initial guess from \eqref{sim_arch_theta0_1}. The boxplots shows the distribution of the one hundred trajectories.}
\label{fig:plt_arch_trajectory_mc_offline}
\end{figure}

Here, in Figure \ref{fig:plt_arch_trajectory_mc_online}, we can see that AdaVol converges to the true parameter values with low sensitivity to the choice of initial values.
Moreover, this convergence occurs within the first few thousand observations.
However, in Figure \ref{fig:plt_arch_trajectory_mc_offline}, we see the opposite in which $\widetilde{\theta}_n$ has convergence issues; it is consistently underestimating the $\omega$ parameter.
Furthermore, the $\alpha_1$ parameter range does not appear to be decreasing over time, and the range seems larger than AdaVol's.

As we observe the true volatility process $(\sigma_t)$ in this section, we can evaluate the predicted volatility processes' accuracy.
We do this using the Mean Percentage Errors (MPE) given as
\begin{align} \label{eq:acc_score_mpe}
\widehat{\sigma}_{\text{MPE}} = \frac{1}{n} \sum_{t=1}^{n} \frac{\sigma_{t} - \widehat{\sigma}_{t}}{\sigma_{t}}
\text{ and }
\widetilde{\sigma}_{\text{MPE}} = \frac{1}{n} \sum_{t=1}^{n} \frac{\sigma_{t} - \widetilde{\sigma}_{t}}{\sigma_{t}},
\end{align}
and the Mean Absolute Percentage Errors (MAPE) given by
\begin{align} \label{eq:acc_score_mape}
\widehat{\sigma}_{\text{MAPE}} = \frac{1}{n} \sum_{t=1}^{n} \frac{\vert \sigma_{t} - \widehat{\sigma}_{t} \vert}{\sigma_{t}}
\text{ and }
\widetilde{\sigma}_{\text{MAPE}} = \frac{1}{n} \sum_{t=1}^{n} \frac{\vert \sigma_{t} - \widetilde{\sigma}_{t} \vert}{\sigma_{t}},
\end{align}
where $(\widehat{\sigma}_{t})$ is coming from AdaVol and $(\widetilde{\sigma}_{t})$ from the IQMLE approximation.
Note that $\widetilde{\sigma}_{t}$'s estimation is the same as for the IQMLE approximation $\widetilde{\theta}_{t}$, i.e., $(\widetilde{\sigma}_{t})_{(k-2000)+1 \leq t \leq k}$ is estimated using $(X_t)_{1 \leq t \leq k}$ for $k=2000,4000, \dots, n$.

In the rest of this section, we will use random parameters to generalize our studies, limiting the potential bias from having fixed parameters (See Section \ref{sec:simulations}).
Our routine is as follows: We draw a random true parameter vector $\theta_{0} \in \K$ from which we generate our observations $(X_{t})$.
Based on these observations $(X_{t})$, we calculate our estimates using (a random) $\widehat{\theta}_{0} = \widetilde{\theta}_{0} \in \K$.
Then, we evaluate our estimates using an accuracy score, e.g., MPE and MAPE. 
Finally, we repeat all these steps the desired number of times.
Boxplots of one hundred accuracy scores, MPE in \eqref{eq:acc_score_mpe} and MAPE in \eqref{eq:acc_score_mape}, can be found in Figure \ref{fig:plt_arch_boxplot_mes_rtrue_rinit}.
In the top graph of Figure \ref{fig:plt_arch_boxplot_mes_rtrue_rinit}, one can observe the MPE (for both methods) is symmetric around zero, but $\widetilde{\sigma}_{\text{MPE}}$ has a negative tail, meaning the iterative method may overestimate the volatility in some cases.
Also, the spread of $\widetilde{\sigma}_{\text{MPE}}$ is higher than the $\widehat{\sigma}_{\text{MPE}}$, which is clearly seen by looking at $\widetilde{\sigma}_{\text{MAPE}}$ in the bottom graph of Figure \ref{fig:plt_arch_boxplot_mes_rtrue_rinit}.

\begin{figure}[h!]
\centering
\includegraphics[width=1.0\linewidth]{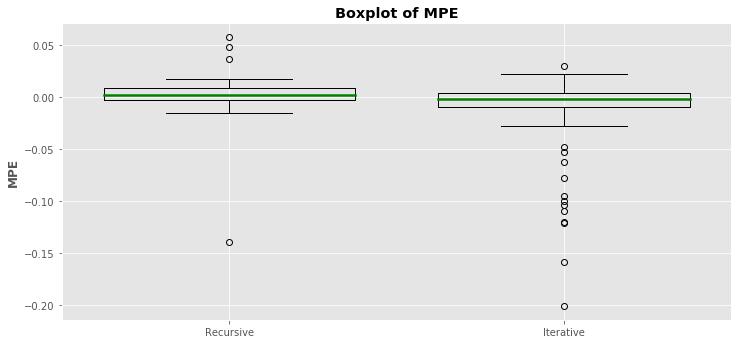}
\includegraphics[width=1.0\linewidth]{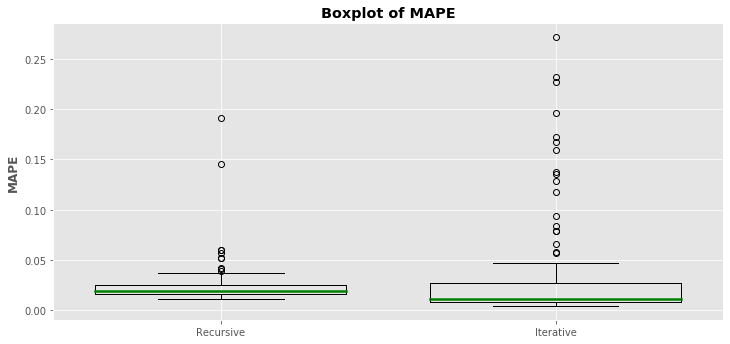}
\caption{Boxplots of one hundred accuracy scores MPE \eqref{eq:acc_score_mpe} and MAPE \eqref{eq:acc_score_mape} using an ARCH$(1)$ process with random true parameter vector and initial guess in $\K$.}
\label{fig:plt_arch_boxplot_mes_rtrue_rinit}
\end{figure}

Another way of measuring the accuracy can be made by studying the conditional quantiles using the recursive $(\widehat{\sigma}_{t})$ and iterative $(\widetilde{\sigma}_{t})$ predicted volatility processes (\cite{biau2011}).
Under the assumption of standard Gaussian innovations, $X_t$ is Gaussian with zero mean and variance $\sigma^2_t$. 
Thus, for any $\alpha \in (0,1)$, the $\alpha$-quantile of a Gaussian distribution $\NORM (0, \sigma^2_t)$ is $\sigma_{t} \Phi^{-1}(\alpha)$, where $\Phi^{-1}(\alpha)$ is the $\alpha$-quantile of the standard Gaussian distribution.
We use the so-called $\alpha$-quantile loss function proposed by \cite{koenker1978}: The $\alpha$-quantile loss function $\rho_{\alpha}$ using the volatility process $\sigma_{t}$ is defined as
\begin{align} \label{eq:quantile_loss_function}
\rho_{\alpha} (X_{t}, \sigma_{t}) = 
\begin{cases}
\alpha \left(X_{t} - \Phi^{-1}(\alpha) \sigma_{t}\right), &\text{for } X_{t} > \Phi^{-1}(\alpha) \sigma_{t}, \\
(1-\alpha) \left( \Phi^{-1}(\alpha) \sigma_{t} - X_{t}\right), &\text{for } X_{t} \leq \Phi^{-1}(\alpha) \sigma_{t},
\end{cases}
\end{align}
with tilting parameter $\alpha \in (0,1)$.
The idea behind the $\alpha$-quantile loss function is to penalize quantiles of low probability more for overestimation than for underestimation (and reversely for high probability quantiles). 
We evaluate across the $\alpha$-quantile scores $\rho_{\alpha}$ of $(\sigma_{t})$ by the (normalized) cumulative $\alpha$-quantile scoring function $QS_{\alpha}$:
\begin{align} \label{eq:cum_quantile_loss_function}
QS_{\alpha} (X_{n}, \sigma_{n}) = \frac{1}{n} \sum_{t=1}^{n} \sum_{m=1}^{M} \rho_{\alpha_m} (X_{t}, \sigma_{t}),
\end{align}
with $M$ as the number of quantiles $\alpha = \{\alpha_1, \dots, \alpha_M\}$.
The lowest $QS_{\alpha}$ score indicates the best ability of volatility forecast.
The findings of one hundred $QS_{\alpha} (X_{n}, \widehat{\sigma}_{n})$ and $QS_{\alpha} (X_{n}, \widetilde{\sigma}_{n})$ scores is presented in Figure \ref{fig:plt_arch_boxplot_ql_rtrue_rinit}, where we have used $\alpha = \{0.01, 0.02, \dots, 0.99\}$, a random true parameter vector and random initialization in $\K$. 
The $QS_{\alpha}$ scores in Figure \ref{fig:plt_arch_boxplot_ql_rtrue_rinit} are indistinguishable.
This indicates no loss of generality in using our recursive method even though our estimates are calculated only once, making them more adaptable over time.
Surprisingly, the iterative method is not superior, even when forward-looking (with up to two thousand observations).

\begin{figure}[h!]
\centering
\includegraphics[width=1.0\linewidth]{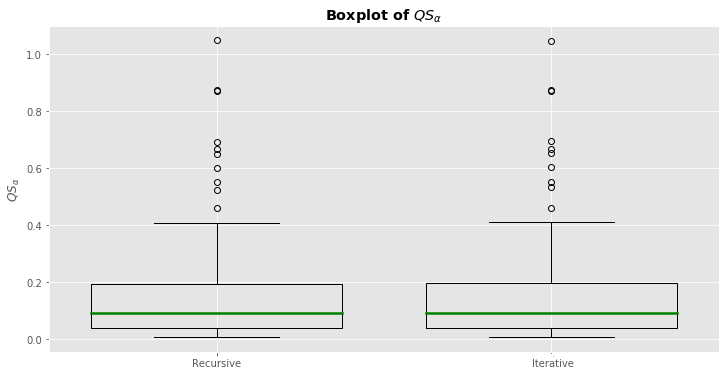}
\caption{Boxplots of one hundred $QS_{\alpha}$ scores with $\alpha = \{0.01, 0.02, \dots, 0.99\}$ using an ARCH$(1)$ model with random true parameter vector and initial value in $\K$.}
\label{fig:plt_arch_boxplot_ql_rtrue_rinit}
\end{figure}


\subsubsection{GARCH Models}


Figure \ref{fig:plt_garch_trajectory_mc_online} and \ref{fig:plt_garch_trajectory_mc_offline} shows the trajectories of the parameter estimates $\widehat{\theta}_{n} = (\widehat{\omega}^{(n)}, \widehat{\alpha}_{1}^{(n)}, \widehat{\beta}_{1}^{(n)})^T$ and $\widetilde{\theta}_{n} = (\widetilde{\omega}^{(n)}, \widetilde{\alpha}_{1}^{(n)}, \widetilde{\beta}_{1}^{(n)})^T$ for a GARCH$(1,1)$ model with the true parameter vector and initial guess given by
\begin{align} \label{sim_garch_theta0_1}
\theta_0 =
 \begin{pmatrix} \omega \\ \alpha_1 \\ \beta_1 \end{pmatrix}
 =
 \begin{pmatrix} 1 \cdot 10^{-8} \\ 0.2 \\ 0.7 \end{pmatrix}
 \text{ and }
\widehat{\theta}_0 
= \widetilde{\theta}_0 
= \begin{pmatrix} 5 \cdot 10^{-8} \\ 0.1 \\ 0.8 \end{pmatrix}.
\end{align}
As for the ARCH$(1)$ model, we observe a lower spread in the parameter trajectories coming from AdaVol $\widehat{\theta}_{n}$ than from the IQMLE approximation $\widetilde{\theta}_{n}$.
Moreover, the iterative $\widetilde{\theta}_{n}$ is consistently overestimating the $\beta_{1}$ parameter (and underestimating the $\alpha_{1}$ parameter), indicating a bias relative to the initial value.
It is worth mentioning that even if all initial values are in the stationary region, i.e., $\widehat{\theta}_{0} = \widetilde{\theta}_{0} = \theta_{0} \in \K$, we still have a proper amount of fluctuation in the parameter trajectories. 
As discussed before, this may partially be due to the volatility introduced by the gradient method and the flatness of the QL loss (\cite{zumbach2000pitfalls}).
Nevertheless, our recursive method possesses a remarkable convergence already after the first few thousand observations.

\begin{figure}[h!]
\centering
\includegraphics[width=1.0\linewidth]{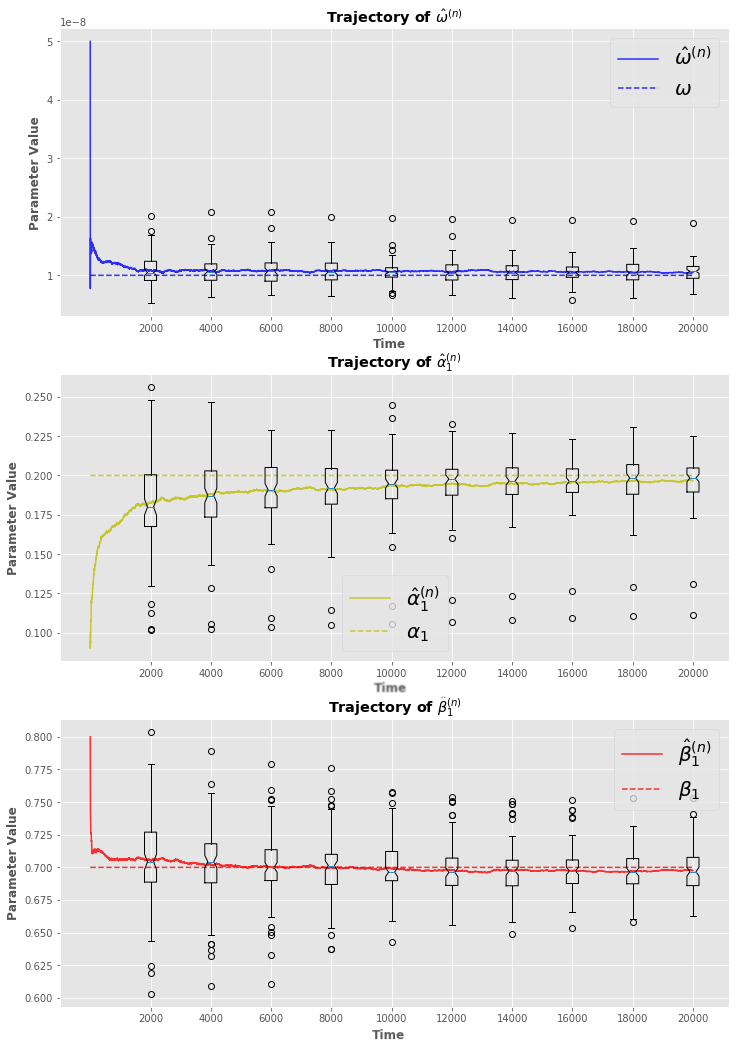}
\caption{Average trajectory (solid line) of one hundred $\widehat{\theta}_n$'s for a GARCH$(1,1)$ process with true parameter vector (dotted line) and initial guess given in \eqref{sim_garch_theta0_1}. The boxplots shows the distribution of the one hundred trajectories.}
\label{fig:plt_garch_trajectory_mc_online}
\end{figure}

\begin{figure}[h!]
\centering
\includegraphics[width=1.0\linewidth]{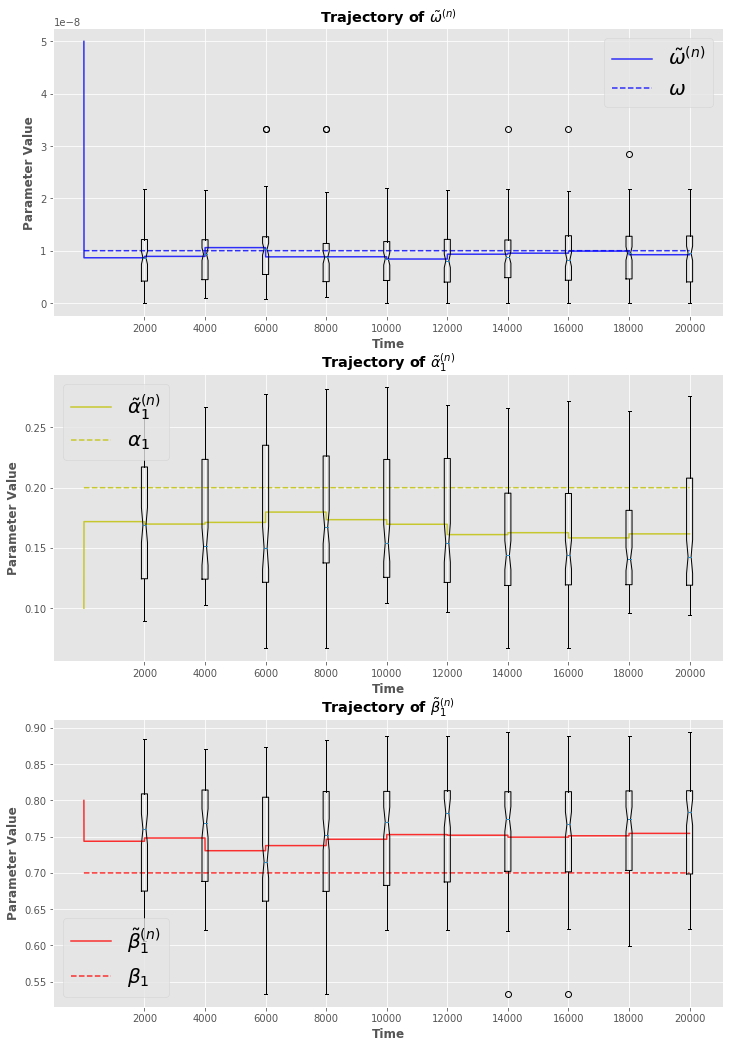}
\caption{Average trajectory (solid line) of one hundred $\widetilde{\theta}_n$'s for a GARCH$(1,1)$ process with true parameter vector (dotted line) and initial guess given in \eqref{sim_garch_theta0_1}. The boxplots shows the distribution of the one hundred trajectories.}
\label{fig:plt_garch_trajectory_mc_offline}
\end{figure}

The accuracy scores, namely MPE from \eqref{eq:acc_score_mpe} and MAPE from \eqref{eq:acc_score_mape}, can be found in Figure \ref{fig:plt_garch_boxplot_mes_rtrue_rinit} for the GARCH$(1,1)$ model using random true parameter vector and random initial values in $\K$.
By comparing our methods using random initializations, we circumvent the possible bias from the initial guess, which we observed in Figure \ref{fig:plt_garch_trajectory_mc_offline} for the iterative method.
As in the ARCH$(1)$ case, we obtain a lower spread for $\widehat{\sigma}_{\text{MPE}}$ than $\widetilde{\sigma}_{\text{MPE}}$.
Nevertheless, one should still expect some probability of ending up with an irregular solution where the AdaVol algorithm fails to converge.

\begin{figure}[h!]
\centering
\includegraphics[width=1.0\linewidth]{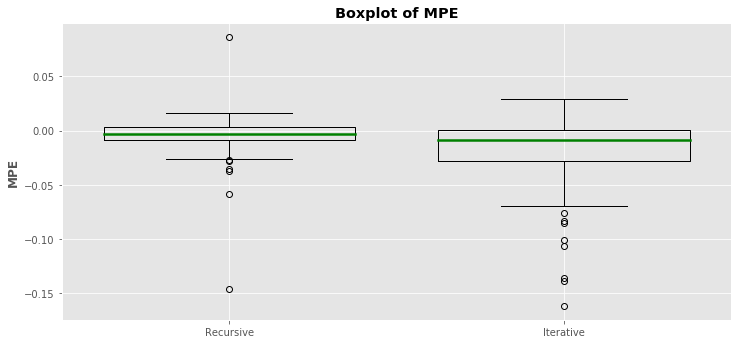}
\includegraphics[width=1.0\linewidth]{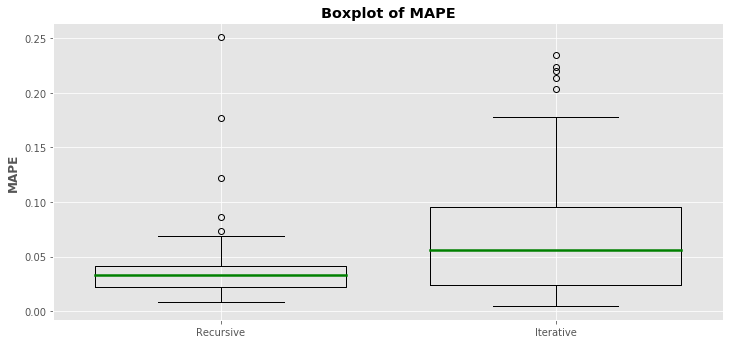}
\caption{Boxplots of one hundred accuracy scores MPE \eqref{eq:acc_score_mpe} and MAPE \eqref{eq:acc_score_mape} using a GARCH$(1,1)$ process with true parameter vector and random initial guess in $\K$.}
\label{fig:plt_garch_boxplot_mes_rtrue_rinit}
\end{figure}

Figure \ref{fig:plt_garch_boxplot_ql_rtrue_rinit} presents the results of one hundred $QS_{\alpha}$ scores with random true parameter vector and initial value in $\K$.
Again, the $QS_{\alpha}$ scores are indistinguishable (even when the iterative method is forward-looking).

\begin{figure}[h!]
\centering
\includegraphics[width=1.0\linewidth]{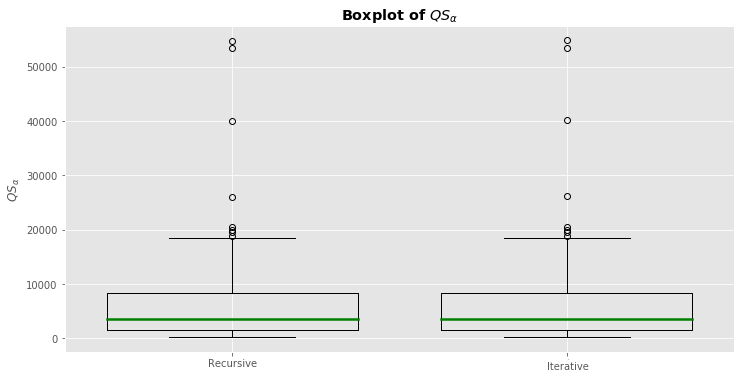}
\caption{Boxplots of one hundred $QS_{\alpha}$ scores with $\alpha = \{0.01, 0.02, \dots, 0.99\}$ using the GARCH$(1,1)$ model with random true parameter vector and initial value in $\K$.}
\label{fig:plt_garch_boxplot_ql_rtrue_rinit}
\end{figure}


\subsection{Real-life Observations}


We will now demonstrate AdaVol's abilities on real-life observations showing how our technique works in practice.
Table \ref{tab:idx_overview} shows an overview of the used stock market indices.
All empirical studies use the GARCH$(1,1)$ model, but higher-order parameters may yield a better fit for some stock market indices.
As the observation period spans over a long time, it is unlikely that the log-return series is stationary.
To exhibit AdaVol's ability to adapt to time-varying estimates, we begin by considering the S\&P500 Index in Section \ref{sec:sp500}. 
Afterward, in Section \ref{sec:idx}, we investigate the remaining six stock market indices presented in Table \ref{tab:idx_overview}, namely the CAC, DAX, DJIA, NDAQ, NKY, and RUT index.

\begin{table}[h!]
\begin{tabular}{llll} \hline
Stock Market Index & & Period \\ \hline
CAC 40 & (CAC) & March 1990 - Sep. 2020 \\
DAX 30 & (DAX) & Jan. 1988 - Sep. 2020 \\
Dow Jones & (DJIA) & Feb. 1985 - Sep. 2020 \\
NASDAQ Composite & (NDAQ) & Feb. 1971 - Sep. 2020 \\
Nikkei 225 & (NKY) & Jan. 1965 - Sep. 2020 \\
Russell 2000 & (RUT) & Nov. 1987 - Sep. 2020 \\ 
Standard \& Poor's 500 & (S\&P500) & Jan. 1950 - Sep. 2020 \\
\hline
\end{tabular}
\centering
\caption{Overview of considered stock market indices including their observation periods. The observations consist of daily log-returns which are defined as log differences of the closing prices of the index between two consecutive days.}
\label{tab:idx_overview}
\end{table}


\subsubsection{Application to the S\&P500 Index} \label{sec:sp500}


We apply our method on the S\&P500 Index from January 1950 to September 2020, consisting of $n=17672$ observations to test real-life data performance.
We employ the GARCH$(1,1)$ model with initial values:
\begin{align} \label{sp500_garch_thetahat0_1}
\widehat{\theta}_0 
= \widetilde{\theta}_0 
= \begin{pmatrix} 5 \cdot 10^{-5} \\ 0.05 \\ 0.9 \end{pmatrix}.
\end{align}
The QML trajectories can be seen in Figure \ref{fig:plt_sp500_garch}.
The produced AdaVol estimates $\widehat{\theta}_n = (\widehat{\omega}^{(n)}, \widehat{\alpha}_{1}^{(n)}, \widehat{\beta}_{1}^{(n)})^T$ experience some fluctuations initially, but as it vaporizes, it is clear that our estimates change over time. 
Most remarkable are the shifts our estimates make around some historical market crashes, e.g., Black Monday, the financial crisis, and COVID-19.
The instant shift in our estimates is an appealing property for detecting structural breaks.
It is noteworthy that the estimates of the IQMLE approximation $\widetilde{\theta}_n = (\widetilde{\omega}^{(n)}, \widetilde{\alpha}_{1}^{(n)}, \widetilde{\beta}_{1}^{(n)})^T$ are predominantly constant over time with minor changes except for some years between $1990$ and $2000$, where we detect a shift to lower $\widetilde{\beta}_{1}^{(n)}$ values and higher $\widetilde{\omega}^{(n)}$ values.

\begin{figure}[h!]
\centering
\includegraphics[width=1.0\linewidth]{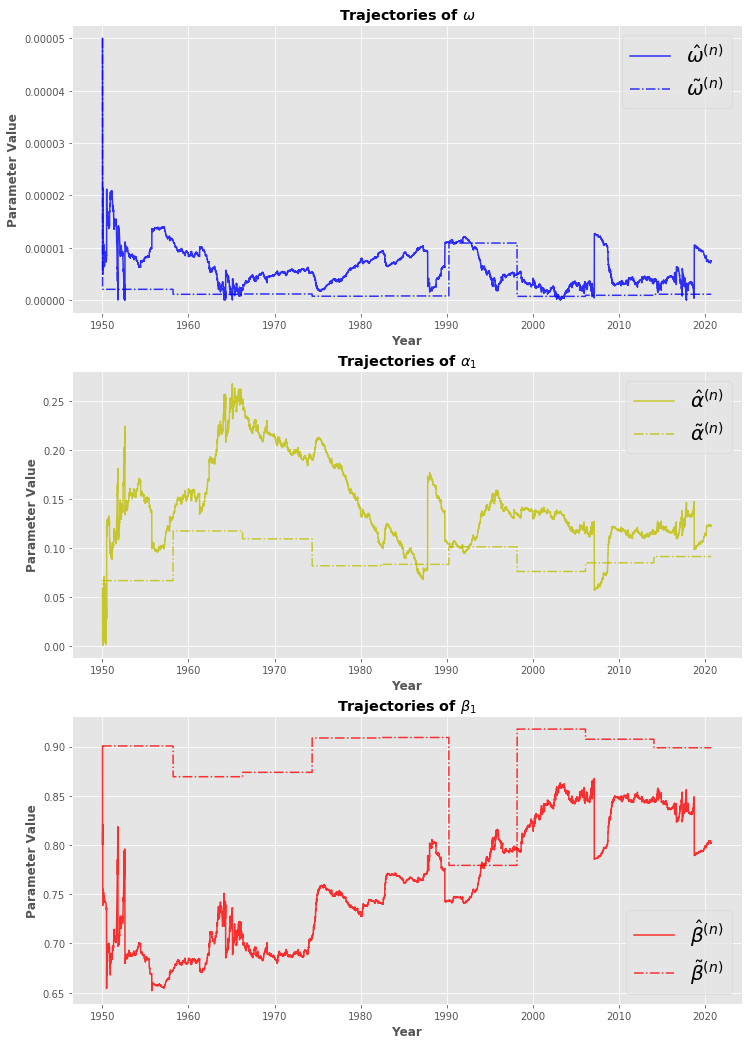}
\caption{Trajectory of the recursive $\widehat{\theta}_n$ (solid line) and iterative $\widetilde{\theta}_n$ (semi-dotted line) QML estimate using a GARCH$(1,1)$ model on S\&P500 Index log-returns from year 1950 to 2020. Both methods use initial value given in \eqref{sp500_garch_thetahat0_1}.}
\label{fig:plt_sp500_garch}
\end{figure}

In Figure \ref{fig:plt_sp500_return_b}, we have the log-returns $r_t$ of the S\&P500 Index, and the confidence intervals $\bar{r}\pm1.96\widehat{\sigma}_{t}$ and $\bar{r}\pm1.96\widetilde{\sigma}_{t}$ using the recursive $\widehat{\sigma}_{t}$ and iterative $\widetilde{\sigma}_{t}$ predicted volatilities, where $\bar{r}$ is the mean of the log-returns $r_t$.
It seems that the recursive method $\widehat{\sigma}_{t}$ adapts more rapidly than the iterative one $\widetilde{\sigma}_{t}$ to changes in the S\&P500 Index observations $r_t$. 
Especially in Figure \ref{fig:plt_sp500_return_b}, under the COVID-19 crisis, we encountered a period with a substantial volatility increase.
Here, we observe $\widehat{\sigma}_{t}$'s ability to track changing volatilities better than $\widetilde{\sigma}_{t}$.

\begin{figure}[h!]
\centering
\includegraphics[width=1.0\linewidth]{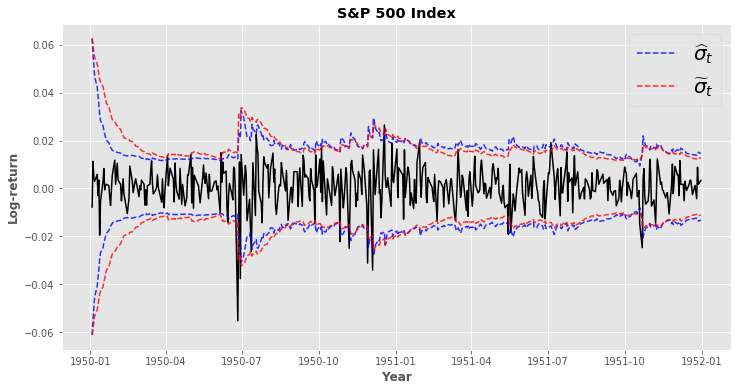}
\includegraphics[width=1.0\linewidth]{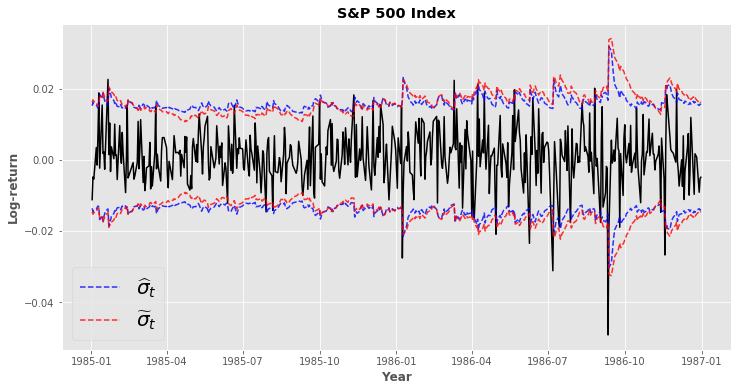}
\includegraphics[width=1.0\linewidth]{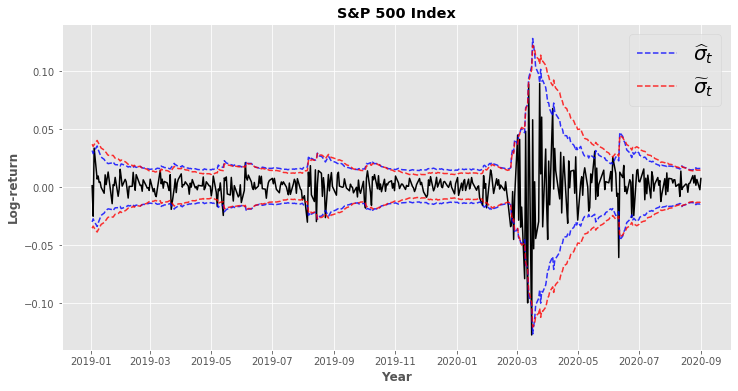}
\caption{Log-returns $r_t$ of S\&P500 Index (solid lines) and confidence intervals $\bar{r}\pm1.96\widehat{\sigma}_{t}$ and $\bar{r}\pm1.96\widetilde{\sigma}_{t}$ (dotted lines) using the recursive $\widehat{\sigma}_{t}$ (blue) and iterative $\widetilde{\sigma}_{t}$ (red) predicted volatilities, where $\bar{r}$ is the mean of the log-returns $r_t$. From top to bottom, we have Jan. 1950 to Jan. 1952, Jan. 1985 to Jan. 1987, and Jan. 2019 to Sep. 2020.}
\label{fig:plt_sp500_return_b}
\end{figure}

In the absence of the true (unobserved) variance process $(\sigma_{t}^{2})$, the efficiency of our recursive $(\widehat{\sigma}_{t})$ and the iterative $(\widetilde{\sigma}_{t})$ volatility can be appraised with the use of the squared log-returns $(r_{t}^{2})$.
We use the Mean Absolute Errors (MAE) defined by
\begin{align} \label{eq:acc_score_mae}
\widehat{\sigma}_{\text{MAE}}^{2} = \frac{1}{n} \sum_{t=1}^{n} \vert r_{t}^{2} - \widehat{\sigma}_{t}^{2} \vert
\text{ and }
\widetilde{\sigma}_{\text{MAE}}^{2} = \frac{1}{n} \sum_{t=1}^{n} \vert r_{t}^{2} - \widetilde{\sigma}_{t}^{2} \vert.
\end{align}
In Table \ref{tab:sp_500_garch_mae}, we consider the MAEs for the same periods used in Figure \ref{fig:plt_sp500_return_b}, including for the full dataset. 
The results in Table \ref{tab:sp_500_garch_mae} confirm our conclusions about Figure \ref{fig:plt_sp500_return_b}; the AdaVol method tracks the volatility better than the iterative method. 

\begin{table}[h!]
\begin{tabular}{lll} \hline
Period & $\widehat{\sigma}_{\text{MAE}}^{2}$ & $\widetilde{\sigma}_{\text{MAE}}^{2}$ \\ \hline
Jan. 1950 - Jan. 1952 & 8.2388 & 8.9049 \\
Jan. 1985 - Jan. 1987 & 7.1214 & 7.4723 \\ 
Jan. 2018 - Sep. 2020 & 26.9205 & 30.4775 \\ \hdashline
Jan. 1950 - Sep. 2020 & 10.1861 & 10.6731 \\ \hline
\end{tabular}
\centering
\caption{MAEs \eqref{eq:acc_score_mae} using log-returns $r_t$ of S\&P500 Index with the recursive $\widehat{\sigma}_{t}$ and iterative $\widetilde{\sigma}_{t}$ predicted volatilities.  Both methods has initial value given in \eqref{sp500_garch_thetahat0_1}. The $\widehat{\sigma}_{\text{MAE}}^{2}$ and $\widetilde{\sigma}_{\text{MAE}}^{2}$ numbers are scaled by $10^{-5}$.}
\label{tab:sp_500_garch_mae}
\end{table}

Figure \ref{fig:sp500_boxplot_ql_rinit} contains the results of one hundred $QS_{\alpha}$ scores using the recursive $(\widehat{\sigma}_{t})$ and iterative $(\widetilde{\sigma}_{t})$ volatility process, respectively, with random initial values in $\K$.
Remarkably, AdaVol outperforms the iterative method, although the latter uses future information, i.e., $(\widetilde{\sigma}_t)_{(k-2000)+1 \leq t \leq k}$ is estimated using $(r_t)_{1 \leq t \leq k}$ for $k=2000,4000, \dots, 16000, 17505$.
This indicates that one could achieve better performance using the recursive method, even if it only predicts volatility using previous information.

\begin{figure}[h!]
\centering
\includegraphics[width=1.0\linewidth]{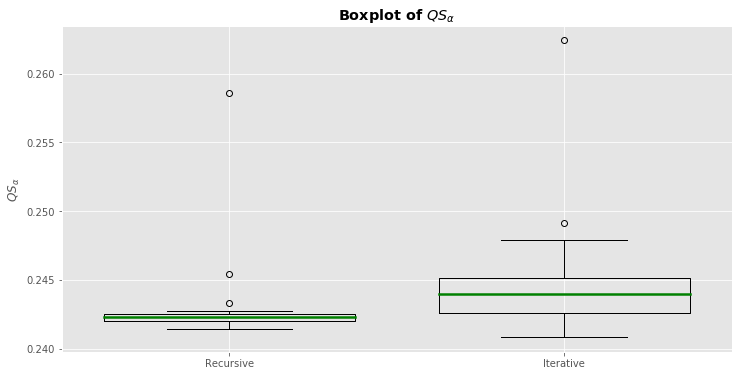}
\caption{Boxplots of one hundred $QS_{\alpha}$ scores with use of the recursive $\widehat{\sigma}_{t}$ and iterative $\widetilde{\sigma}_{t}$ volatility process, respectively, for $\alpha = \{0.01, 0.02, \dots, 0.99\}$, using the GARCH$(1,1)$ model on the log-returns $r_t$ of S\&P500 Index with random initial value in $\K$.}
\label{fig:sp500_boxplot_ql_rinit}
\end{figure}


\subsubsection{Other Stock Market Indices} \label{sec:idx}


We now extend our analysis to the remaining stock market indices from Table \ref{tab:idx_overview}, namely the CAC, DAX, DJIA, NDAQ, NKY, and RUT index.
In Figure \ref{fig:idx_return_b}, we can observe AdaVol's ability to adapt to time-varying parameters seems to hold for several stock market indices.
These figures show a clear benefit in recursive estimation as it increases adaptivity that may be advantageous under a financial crisis such as the COVID-19.

\begin{figure*}[h!]
\centering
\begin{tabular}{@{}cc@{}}
\includegraphics[width=0.5\linewidth]{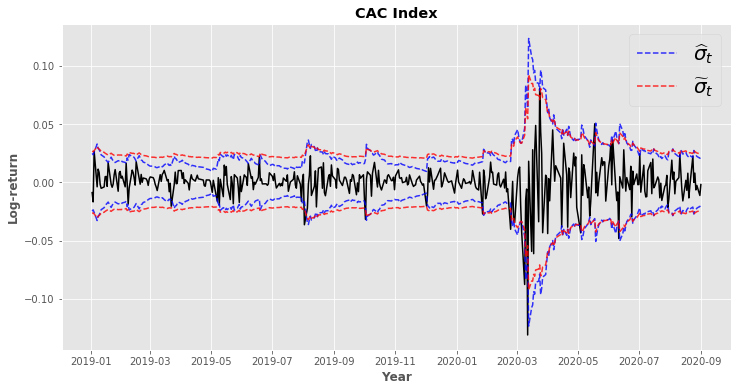} &
\includegraphics[width=0.5\linewidth]{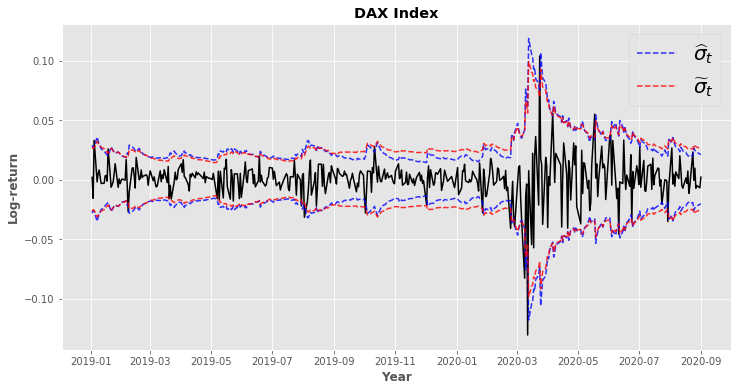} \\
\includegraphics[width=0.5\linewidth]{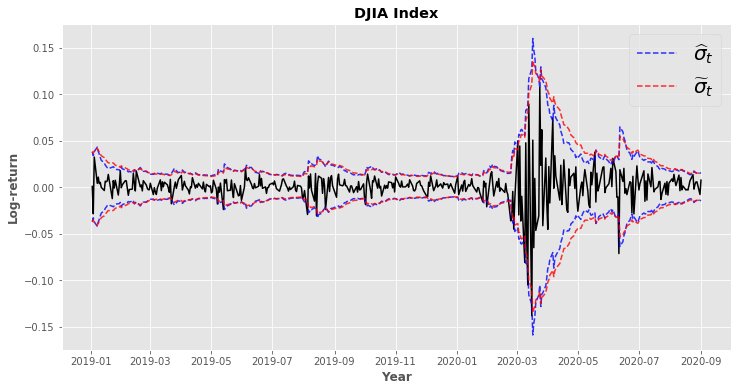} &
\includegraphics[width=0.5\linewidth]{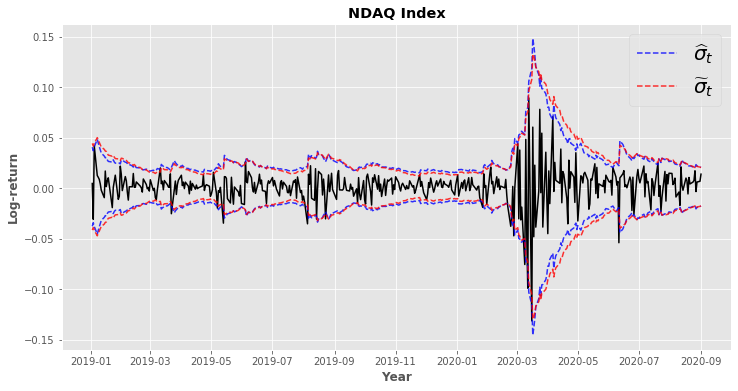} \\
\includegraphics[width=0.5\linewidth]{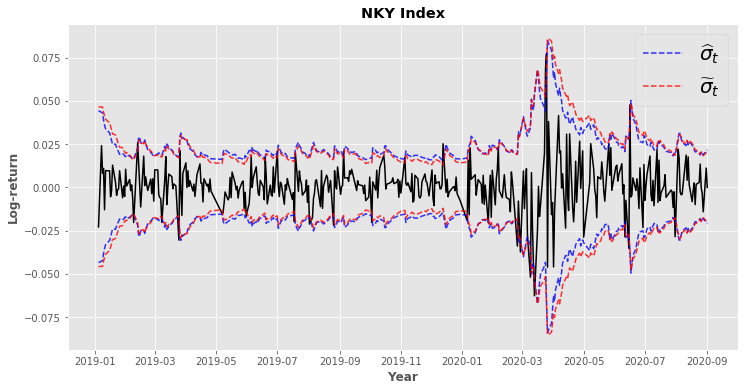} &
\includegraphics[width=0.5\linewidth]{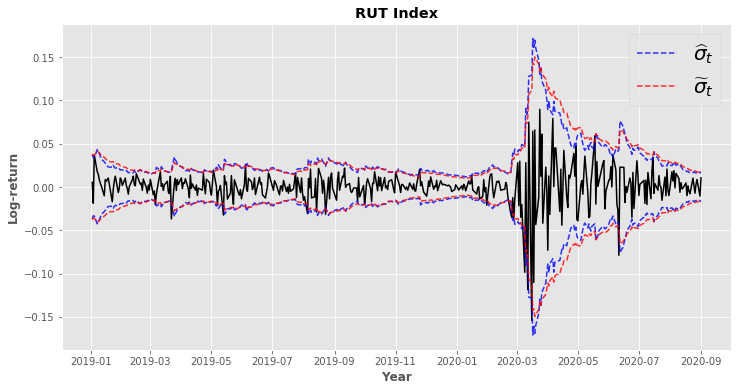}
\end{tabular}
\caption{Log-returns $r_t$ of the CAC (top-left), DAX (top-right), DJIA (mid-left), NDAQ (mid-right), NKY (bottom-left) and RUT (bottom-right) index (solid lines) and confidence intervals $\bar{r}\pm1.96\widehat{\sigma}_{t}$ and $\bar{r}\pm1.96\widetilde{\sigma}_{t}$ (dotted lines) using the recursive $\widehat{\sigma}_{t}$ (blue) and iterative $\widetilde{\sigma}_{t}$ (red) predicted volatilities, where $\bar{r}$ is the mean of the log-returns $r_t$. The period is Jan. 2019 to Sep. 2020.}
\label{fig:idx_return_b}
\end{figure*}

These conclusions are confirmed in Figure \ref{fig:idx_boxplot_ql_rinit}, where we have one hundred $QS_{\alpha}$ scores using the recursive $(\widehat{\sigma}_{t})$ and iterative $(\widetilde{\sigma}_{t})$ volatility process with random initial values in $\K$.
As for the S\&P500 Index (in Figure \ref{fig:sp500_boxplot_ql_rinit}), our findings indicate that the recursive approach estimates the $QS_{\alpha}$ quantiles better than the iterative method, both on average and with a lower spread.

The assumption of having an underlying data generation process with constant "true" parameters may not hold in real-life examples.
Thus, AdaVol seems to have an advantage compared to the iterative method, as it estimates the parameters step-by-step. 
In contrast, the iterative method always has to estimate the parameters using all observations over an extensive period of time.

\begin{figure*}[h!]
\centering
\begin{tabular}{@{}cc@{}}
\includegraphics[width=0.5\linewidth]{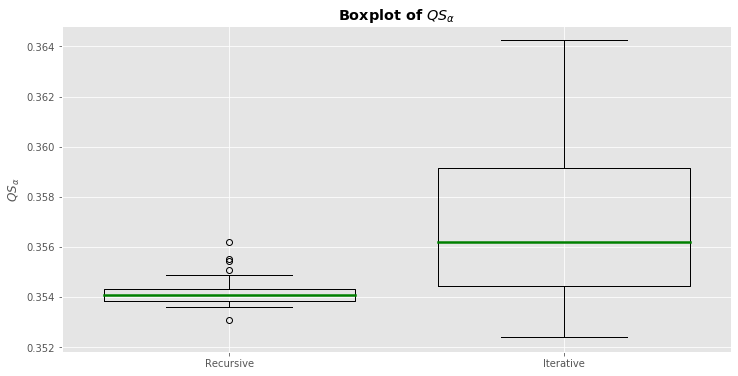} &
\includegraphics[width=0.5\linewidth]{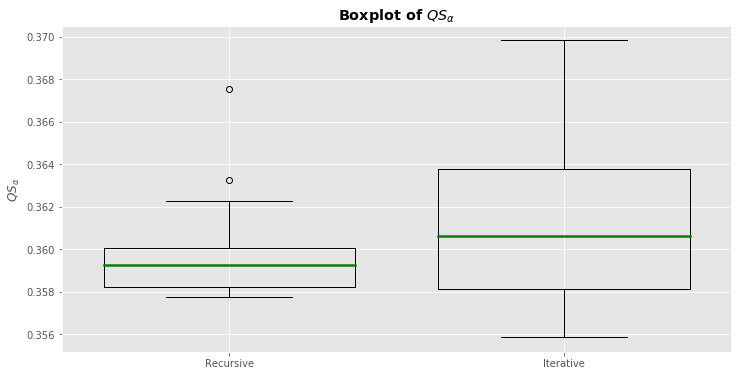} \\
\includegraphics[width=0.5\linewidth]{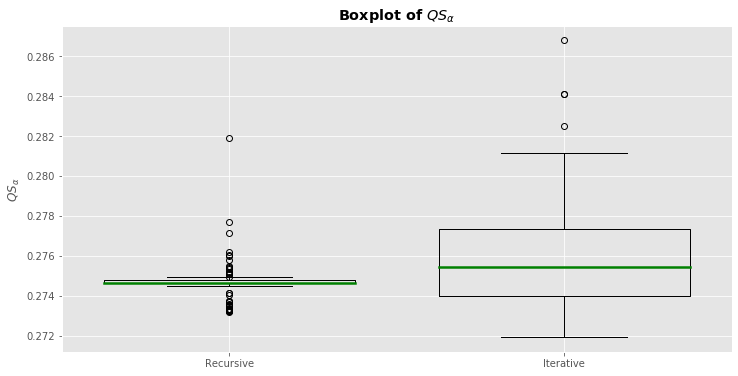} &
\includegraphics[width=0.5\linewidth]{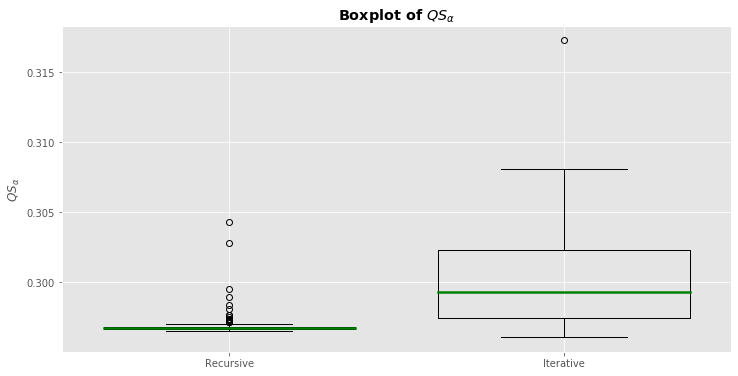} \\
\includegraphics[width=0.5\linewidth]{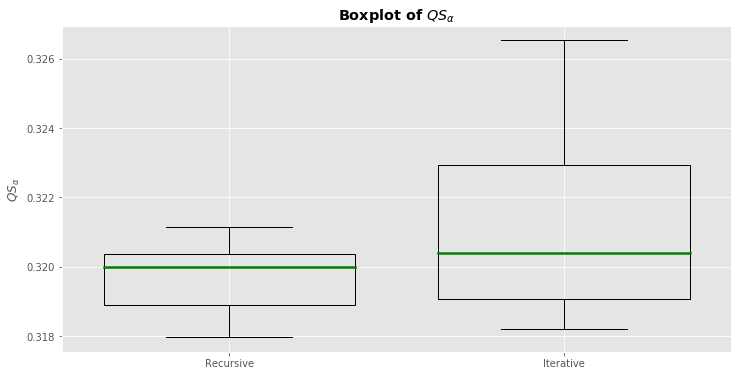} &
\includegraphics[width=0.5\linewidth]{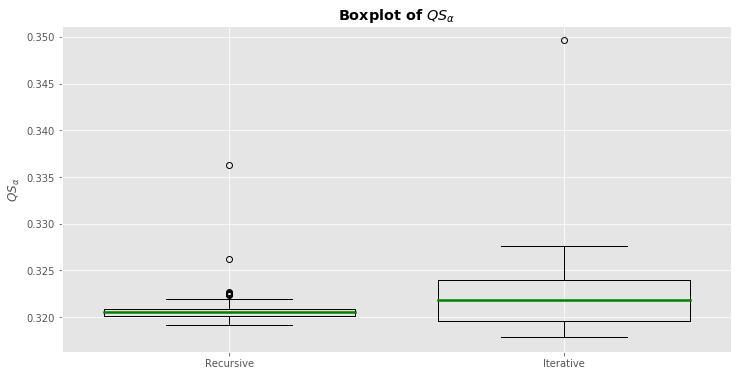}
\end{tabular}
\caption{Boxplots of one hundred $QS_{\alpha}$ scores with the use of the recursive $\widehat{\sigma}_{t}$ and iterative $\widetilde{\sigma}_{t}$ volatility process, respectively, for $\alpha = \{0.01, 0.02, \dots, 0.99\}$, using the GARCH$(1,1)$ model on the log-returns $r_t$ of the CAC (top-left), DAX (top-right), DJIA (mid-left), NDAQ (mid-right), NKY (bottom-left) and RUT (bottom-right) index with random initial values in $\K$.}
\label{fig:idx_boxplot_ql_rinit}
\end{figure*}


\section{Discussion} \label{sec:discussion}


We proved asymptotic local convexity of the QL function in general conditionally heteroscedastic time series models of multiplicative form. 
An interesting question arises: can one prove Theorem \ref{THM_Convex} for a bounded set of $N$ observations? 
Expressed differently, can one find a $N$ bounded, such that we have convergence/convexity of recursive algorithms, e.g., for the GARCH, EGARCH, and AGARCH models. 
To our knowledge, this has not been proved yet.

We proposed an adaptive approach to recursively estimate GARCH model parameters in a streaming setting using the VTE technique (AdaVol). 
AdaVol's design showed to produce resilient and adaptive estimates in our empirical investigations.
The adaptation to time-varying parameters was a surprising advantage that appeared when we applied our method to real-life observations. 
As the assumption of having constant estimates seems not to be the case for the stock indices we analyze, then it is beneficial to have the ability to adapt.
One could facilitate this ability more by incorporating a rolling volatility estimation of $\gamma$ instead of using the sample volatility. 
Combining this with a different learning rate than AdaGrad, which enables continuous learning (e.g., ADAM by \cite{kingma2015}), could encourage adaptability.

The stability of using our recursive approach to solve the QML problem could be improved by using a mini-batch approach. 
A mini-batch approach will lower each incremental volatility as one uses more observations per recursion to update the QML estimate. 
Applying a mini-batch method does not require much more computational power than the stochastic gradient descent, only $\bigO(bd)$, where $b$ is the number of observations used in each (mini-batch) recursion.
Using more observations, we could achieve more consistency and smoothness in the estimation procedure's convergence while keeping favorable computational costs.

Furthermore, an accelerated convergence of our estimates could be obtained by recursion averaging, also called Polyak-Ruppert averaging, which is guaranteed under fairly relaxed conditions (\cite{polyak1992,ruppert1988}).
This Polyak-Ruppert average estimate could be utilized solely or employed as a benchmark to detect structural breaks.


\section*{Acknowledgement}


This work was supported by a grant from R\'egion Ile de France.


\appendix


\section{Proofs} \label{sec:appendix:proofs}


\begin{proof}[Proof of Theorem \ref{THM_Convex}]
To prove local strong convexity for the approximate QL function $\widehat{L}_n$ using the approximate QMLE $\widehat{\theta}_n^{*}$, we first list some bounds for the Hessians: under the regularity conditions on the derivatives of $h_t$, then using \eqref{ONE_LIKE_HAT}, we can write
\begin{align*}
\nabla_{\theta} l_t (\theta) 
=& \frac{1}{2} \frac{\nabla_{\theta} h_t(\theta)}{h_t(\theta)} \left( 1- \frac{X_t^2}{h_t(\theta)} \right)
\end{align*}
and
\begin{align*}
\nabla_{\theta}^{2} l_t (\theta) 
=& \frac{1}{2h_t^2(\theta)} \Bigg( \nabla_{\theta} h_t(\theta)^T \nabla_{\theta} h_t(\theta) \left(  \frac{2X_t^2}{h_t(\theta)} -1 \right) 
+ \nabla_{\theta}^{2} h_t(\theta) \left( h_t(\theta) - X_t^2 \right) \Bigg),
\end{align*}
where the Hessian $H_n(\theta)$ is defined as $n^{-1}\nabla_{\theta}^{2} L_{n} \left(\theta \right) = n^{-1} \sum_{t=1}^{n} \nabla_{\theta}^{2} l_t (\theta)$. Similarly, for $\nabla_{\theta} \widehat{l}_t (\theta)$, $\nabla_{\theta}^{2} \widehat{l}_t (\theta)$, and $\widehat{H}_n (\theta)$, we replace $h_t(\theta), \nabla_{\theta} h_t(\theta)$ and $\nabla_{\theta}^{2} h_t(\theta)$ by $\widehat{h}_t(\theta),\nabla_{\theta} \widehat{h}_t(\theta)$ and $\nabla_{\theta}^{2} \widehat{h}_t(\theta)$, respectively. 
From Assumption \ref{assump:cond_bounds}, we know $n^{-1}\norm{\nabla_{\theta}^{2} \widehat{L}_n  - \nabla_{\theta}^{2} L_n }_\K \overset{\text{a.s.}}{\longrightarrow} 0$ for $n \rightarrow \infty$. Hence, for some random $N_{1}$ large enough, there exists $\eps>0$ such that $n^{-1} \norm{ \nabla_{\theta}^{2} \widehat{L}_n  - \nabla_{\theta}^{2} L_n}_\K < \eps$ for all $n \geq N_{1}$ a.s. 
As a consequence, we get
\begin{align} \label{HES_NORM_1}
\norm{ \widehat{H}_n - H_n}_{\K}  < \eps, \quad \text{a.s.},
\end{align}
for all $n \geq N_{1}$.
Similarly, applying the ergodic theorem on the  integrable sequence (uniformly over $\K$) $(\nabla_{\theta}^{2} l_t)$ of continuous functions over the compact set $\K$, we obtain $\norm{n^{-1} \sum_{t=1}^n \nabla_{\theta}^{2} l_t   -  \E[ \nabla_{\theta}^{2} l_0 ] }_\K \overset{\text{a.s.}}{\longrightarrow } 0$ for $n \rightarrow \infty$. Then there exists $N_{2}$ such that 
\begin{align} \label{HES_NORM_2}
\norm{  H_n - H_0 }_{\K} < \eps,  \quad \text{a.s.},
\end{align}
for all $n \geq N_{2}$.
Thus, by equation \eqref{HES_NORM_1} and \eqref{HES_NORM_2}, we know there exists $N = \max(N_{1},N_{2})$ such that for all $n \geq N$, we have
\begin{align*}
\norm{\widehat{H}_n - H_0 }_{\K}
\leq \norm{\widehat{H}_n - H_n }_{\K} + \norm{H_n - H_0 }_{\K} 
< 2 \eps, \quad \text{a.s.}
\end{align*}
Especially, as $\norm{\widehat{H}_n - H_0 }_{\K}$ is defined as $\sup_{\theta \in \K} \norm{\widehat{H}_n(\theta) - H_0(\theta) }_{op}$, then
\begin{align} \label{eq:hes_op_bound}
\norm{\widehat{H}_n(\theta) - H_0(\theta) }_{op} < 2 \epsilon,
\end{align}
for all $\theta \in \K$.

From \cite[Lemma 7.2]{straumannmikosch2007}, the asymptotic Hessian $H_0 (\theta_0) = \E [\nabla_{\theta}^{2} l_0 (\theta_0)]$ is a symmetric positive definite matrix a.s. under Assumption \ref{assump:cond_comp_inde}.
As $H_0 (\theta)$ is the limit of the continuous matrix-valued function $H_n(\theta)$, it is itself a continuous matrix-valued function. 
Thus, the eigenvalue function $\lambda_{0}^{i}(\theta)$ for $1\leq i \leq d$ of $H_0 (\theta)$ is also continuous.
The eigenvalues $\lambda_{0}^{i}(\theta_0)$ are positive real numbers with the smallest one $\lambda_{0}^{\min}(\theta_0)$ denoted by
\begin{align*}
\lambda_0^{\min}(\theta_0) = \min_{1\leq i \leq d} \lambda_{0}^{i} (\theta_0) >0,
\end{align*}
satisfying $g^T H_0 (\theta_0) g \geq \lambda_{0}^{\min}(\theta_0) g^Tg$ for all  $g \in \R^{d}\setminus \{ 0\}$.

To shorten the notation, we write with no ambiguity
$H_0 (\theta_0) \succeq \lambda_{0}^{\min} (\theta_0) I_d$ where $I_d$ denotes the $d$-dimensional identity matrix.
By continuity, $\lambda_{0}^{\min} (\theta)$ is positive on a neighborhood $B(\theta_0,\delta)$ such there exist $\eps>0$ satisfying $\lambda_{0}^{\min} (\theta_0) - \eps>0$, meaning
\begin{align*}
H_0(\theta) \succeq (\lambda_{0}^{\min} (\theta_0) - \eps)I_d, 
\end{align*} 
for $\theta \in B(\theta_0,\delta)$.
Hence, for $\theta \in B(\theta_0,\delta)$ and $g \in \R^{d} \setminus \{ 0\}$, we have
\begin{align*}
\frac{g^T \widehat{H}_n (\theta)}{g^Tg} 
&= \frac{g^T H_0 (\theta) g}{g^Tg}
 + \frac{g^T \left( \widehat{H}_n (\theta) - H_0(\theta) \right) g}{g^Tg}
\\ &\geq \lambda_{\min} - \eps - \frac{g^T \norm{ \widehat{H}_n (\theta) - H_0 (\theta) }_{op}  g}{g^Tg} 
\\ &> \lambda_{\min} - 3\eps
\\ &>C, \quad \text{a.s.,}
\end{align*}
using \eqref{eq:hes_op_bound} for all $n \geq N$ by taking $0<\eps < 6^{-1} \lambda_{\min}$ and letting $C = 2^{-1} \lambda_{\min}$. 
Then we have the desired inequality \eqref{CONVEX_HESSIAN_QLIK}.
\end{proof}

\begin{proof}[Proof of Corollary \ref{COR_Convex}]
The uniqueness of the QMLE $\widehat{\theta}_n^{*}$ follows from a Pfanzagl argument (\cite{pfanzagl1969}). 
By Theorem \ref{THM_Convex}, we know there exists $N$ such that
\begin{align*}
\inf_{\theta \in B(\theta_0,\delta_0)} g^T \widehat{H}_n(\theta) g > C g^Tg, \quad \text{a.s.,}
\end{align*}
for all $n \geq N$ where $B(\theta_0, \delta_0)$ denotes the open ball around $\theta_0$ with radius $\delta_0>0$.
For each element $\theta_i \in \K$, we make an open ball $B(\theta_i, \delta_i)$ for $\delta_i>0$ such that the union of $B (\theta_i, \delta_i)$ for all $i$ only contains $\theta_0$ once, i.e., $\theta_0 \notin B (\theta_i,\delta_i)$ for $i \neq 0$. 
As $\K$ is compact and contained in the union of all $B (\theta_i, \delta_i)$, then there is a finite covering of $\K$, i.e., $\K \subseteq \bigcup_{i=0}^k B (\theta_i , \delta_i)$.
Let $\K' = \K \setminus B (\theta_0 , \delta_0)$. 
As $\K'$ is compact, the minimum of the continuous QL function $\E [l_0]$ exists. 
Moreover, as $\E [l_0]$ is a unique minimum at $\theta_0$ under Assumption \ref{assump:cond_exists_solution}, we get
\begin{align*}
\inf_{\theta \in  \K' } \E [l_0 (\theta)] >  \E [l_0 (\theta_0)], \quad \text{a.s.}
\end{align*}
From Assumption \ref{assump:cond_bounds}, we know that $\norm{n^{-1} \widehat{L}_n - L_0}_{\K'} \overset{\text{a.s.}}{\longrightarrow} 0$ as $n \rightarrow \infty$. 
Hence, we have
\begin{align*}
\inf_{\theta \in \K'} n^{-1} \widehat{L}_n(\theta) \overset{\text{a.s.}}{\longrightarrow} \inf_{\theta \in \K'} L_0(\theta),
\end{align*}
where $\inf_{\theta \in \K'} L_0(\theta) > \E [l_0 (\theta_0)]$.
Thus, the $B(\theta_0, \delta_0)$ gives us a unique global minimum of the QL function $\widehat{L}_n$, i.e.,
\begin{align*}
\inf_{\theta \in \K} n^{-1} \widehat{L}_n(\theta) \geq \E [l_0(\theta_0)], \quad \text{a.s.,}
\end{align*}
where equality only is attained when $\theta = \theta_0$.
\end{proof}


\section{Relative Speed Comparison} \label{sec:appendix:speed}


It is argued that the recursive procedure AdaVol is computationally advantageous as it only processes observations once.
In order to illustrate this advantage, a relative computational speed comparison as in \cite{sucarrat2020} is presented. 
The code is not optimized; it is solely for illustration purposes.
In the streaming data framework, the parameters are estimated recursively as described in Section \ref{sec:applications}.
Meaning, for each $t$, the iterative estimate $\widetilde{\theta}_{t}$ is estimated using the observations $(X_{i})_{1 \leq i \leq t}$ and the previous iterative estimate $\widetilde{\theta}_{t-1}$ as initialization.

An ARCH$(1)$, GARCH$(1,1)$, and GARCH$(2,2)$ model is considered for the computational speed analysis.
Table \ref{tab:relative_speed} shows the relative speed comparison for these models with sample sizes $n=1000$ and $n=2000$.
The overall conclusion is that the AdaVol procedure is faster than the iterative one, e.g., the iterative estimation of a GARCH$(1,1)$ model is about $205$ times slower.
Another important observation is the relative speed for different sample sizes $n$, namely, the larger the sample size $n$, the greater the relative speed gain is for the AdaVol procedure.

\begin{table}[h!]
\centering
\begin{tabular}{lccc}
\hline
Model & $n$ & AdaVol & \texttt{arch} \\ \hline
ARCH$(1)$ & $1000$ & $1.00$ & $163.64$ \\
& $2000$ & $1.00$ & $190.12$ \\
GARCH$(1,1)$ & $1000$ & $1.00$ & $204.89$ \\
& $2000$ & $1.00$ & $233.86$ \\
GARCH$(2,2)$ & $1000$ & $1.00$ & $322.33$ \\
& $2000$ & $1.00$ & $328.50$ \\
\hline
\end{tabular}
\caption{Relative speed comparison between AdaVol (\cite{werge2019}) and \texttt{arch} version $4.15$ (\cite{sheppard2020}).
A value of $1.00$ means the method is the fastest. 
A value of $163.64$ means the estimation time of the method is $163.64$ times larger than the fastest.}
\label{tab:relative_speed}
\end{table}


\bibliography{bibliography}
\bibliographystyle{elsarticle-harv} 


\end{document}